\documentclass[10pt]{article}

\usepackage[margin=2cm]{geometry}

\usepackage[colorlinks=true,linkcolor=black,citecolor=black,plainpages=false,pdfpagelabels]{hyperref}
\usepackage{appendix}

\usepackage{times,graphicx,epic,eepic,epsfig,amsmath,latexsym,amssymb,verbatim,revsymb}

\newcommand{\extra}[1]{}

\renewcommand{\comment}[1]{}

\usepackage{color}

\usepackage{amsmath,amssymb,amsfonts}
\usepackage{amsthm}

\newtheorem{theorem}{Theorem}[section]

\newtheorem{corollary}[theorem]{Corollary}

\newtheorem{lemma}[theorem]{Lemma}

\newtheorem{proposition}[theorem]{Proposition}

\theoremstyle{remark}

\def\squareforqed{\hbox{\rlap{$\sqcap$}$\sqcup$}}
\def\qed{\ifmmode\squareforqed\else{\unskip\nobreak\hfil
\penalty50\hskip1em\null\nobreak\hfil\squareforqed
\parfillskip=0pt\finalhyphendemerits=0\endgraf}\fi}
\def\endenv{\ifmmode\;\else{\unskip\nobreak\hfil
\penalty50\hskip1em\null\nobreak\hfil\;
\parfillskip=0pt\finalhyphendemerits=0\endgraf}\fi}

\renewenvironment{proof}{\noindent \textbf{{Proof~} }}{\qed\medskip}
\newenvironment{proof+}[1]{\noindent \textbf{{Proof #1~} }}{\qed\medskip}

\mathchardef\ordinarycolon\mathcode`\:
\mathcode`\:=\string"8000
\def\vcentcolon{\mathrel{\mathop\ordinarycolon}}
\begingroup \catcode`\:=\active
  \lowercase{\endgroup
  \let :\vcentcolon
  }

\newcommand{\nc}{\newcommand}
\nc{\rnc}{\renewcommand}
\nc{\beq}{\begin{equation}}
\nc{\eeq}{{\end{equation}}}
\nc{\beqa}{\begin{eqnarray}}
\nc{\eeqa}{\end{eqnarray}}
\nc{\lbar}[1]{\overline{#1}}
\nc{\bra}[1]{\langle#1|}
\nc{\ket}[1]{|#1\rangle}
\nc{\ketbra}[2]{|#1\rangle\!\langle#2|}
\nc{\braket}[2]{\langle#1|#2\rangle}
\nc{\proj}[1]{| #1\rangle\!\langle #1 |}
\nc{\avg}[1]{\langle#1\rangle}
\nc{\smfrac}[2]{\mbox{$\frac{#1}{#2}$}}
\nc{\tr}{\operatorname{tr}}
\nc{\tracedist}[1]{\Delta_{}\!\left( #1 \right)}
\nc{\fid}[1]{F\!\left( #1 \right)}

\nc{\ox}{\otimes}
\nc{\dg}{\dagger}
\nc{\dn}{\downarrow}
\nc{\cA}{{\cal A}}
\nc{\cB}{{\cal B}}
\nc{\cC}{{\cal C}}
\nc{\cD}{{\cal D}}
\nc{\cE}{{\mathcal E}}
\nc{\cF}{{\cal F}}
\nc{\cG}{{\cal G}}
\nc{\cH}{{\cal H}}
\nc{\cI}{{\cal I}}
\nc{\cJ}{{\cal J}}
\nc{\cK}{{\cal K}}
\nc{\cL}{{\cal L}}
\nc{\cM}{{\cal M}}
\nc{\cN}{{\cal N}}
\nc{\cO}{{\cal O}}
\nc{\cP}{{\cal P}}
\nc{\cR}{{\cal R}}
\nc{\cS}{{\cal S}}
\nc{\cT}{{\cal T}}
\nc{\cU}{{\cal U}}
\nc{\cV}{{\cal V}}
\nc{\cX}{{\cal X}}
\nc{\cZ}{{\cal Z}}

\nc{\entI}{{\bf I}}
\nc{\entIarrow}{{\bf I}^{\leftarrow}}
\nc{\entH}{{\bf H}}
\nc{\entHtwo}{{\mathrm{H}}_2}
\nc{\entS}{{\bf S}}
\nc{\entHmin}{\mathrm{H}_{\min}}
\nc{\aentHmin}{\hat{\mathrm{H}}_{\min}}

\nc{\supp}{\textrm{supp}}

\nc{\entF}{{\bf E}_f}

\nc{\isom}{\simeq}

\nc{\rank}{\operatorname{rank}}
\nc{\rar}{\rightarrow}
\nc{\lrar}{\longrightarrow}
\nc{\polylog}{\operatorname{polylog}}
\nc{\poly}{\operatorname{poly}}
\nc{\1}{{\openone}}

\nc{\weight}{\textbf{w}}
\nc{\hamdist}{d_{H}}

\def\e{\epsilon}

\nc{\Sp}{{{\mathbb S}}}
\nc{\RR}{{{\mathbb R}}}
\nc{\CC}{{{\mathbb C}}}
\nc{\FF}{{{\mathbb F}}}
\nc{\NN}{{{\mathbb N}}}
\nc{\ZZ}{{{\mathbb Z}}}
\nc{\PP}{{{\mathbb P}}}
\nc{\QQ}{{{\mathbb Q}}}
\nc{\UU}{{{\mathbb U}}}
\nc{\OO}{{{\mathbb O}}}
\nc{\EE}{{{\mathbb E}}}
\nc{\id}{{\operatorname{id}}}

\nc{\qubitchannel}{\id_2}
\nc{\bitchannel}{\overline{\id}_2}

\nc{\be}{\begin{equation}}
\nc{\ee}{{\end{equation}}}
\nc{\bea}{\begin{eqnarray}}
\nc{\eea}{\end{eqnarray}}
\nc{\<}{\langle}
\rnc{\>}{\rangle}
\nc{\Hom}[2]{\mbox{Hom}(\CC^{#1},\CC^{#2})}
\nc{\rU}{\mbox{U}}

\nc{\ob}[1]{#1}

\newcommand{\eqdef}	{\stackrel{\textrm{def}}{=}}

\newcommand{\ex}[1]	{\mathbf{E}\left\{ #1 \right\}}
\newcommand{\exc}[2]	{\underset{#1}{\mathbf{E}}\left\{ #2 \right\}}

\newcommand{\pr}[1]	{\mathbf{P}\left\{ #1 \right\}}

\newcommand{\event}[1]	{\left[ #1 \right]}
\newcommand{\eventfont}[1]	{\textsf{#1}}

\renewcommand{\exp}[1]	{\operatorname{exp}\left( #1 \right)}

\newcommand{\ceil}[1]	{\left\lceil #1 \right\rceil}

\nc{\unif}{\textrm{unif}}

\nc{\circuit}{\textrm{circ}}
\nc{\haar}{\textrm{haar}}
\nc{\clifford}{\textrm{clifford}}

\nc{\secondmoment}{\operatorname{M}}
\nc{\Mclifford}{\operatorname{M}_{\clifford}}
\nc{\Mcirc}{\operatorname{M}_{\circuit}}

\nc{\depth}{\operatorname{depth}}

\nc{\binent}{h}

\newcommand{\mfS}{\mathfrak{S}}

\begin{document}

\title{Decoupling with random quantum circuits}

\author{Winton Brown\thanks{D\'{e}partement de Physique, Universit\'{e} de Sherbrooke} \and Omar Fawzi\thanks{Institute for Theoretical Physics, ETH Z\"{u}rich}}

\date{\today}

\maketitle

\begin{abstract}
Decoupling has become a central concept in quantum information theory with applications including proving coding theorems, randomness extraction and the study of conditions for reaching thermal equilibrium. However, our understanding of the dynamics that lead to decoupling is limited. In fact, the only families of transformations that are known to lead to decoupling are (approximate) unitary two-designs, i.e., measures over the unitary group which behave like the Haar measure as far as the first two moments are concerned.  Such families include for example random quantum circuits with $O(n^2)$ gates, where $n$ is the number of qubits in the system under consideration. In fact, all known constructions of decoupling circuits use $\Omega(n^2)$ gates.

Here, we prove that random quantum circuits with $O(n \log^2 n)$ gates satisfy an essentially optimal decoupling theorem. In addition, these circuits can be implemented in depth $O(\log^3 n)$. This proves that decoupling can happen in a time that scales polylogarithmically in the number of particles in the system, provided all the particles are allowed to interact. Our proof does not proceed by showing that such circuits are approximate two-designs in the usual sense, but rather we directly analyze the decoupling property.
\end{abstract}

\section{Introduction}
Consider an observer $E$ that holds some information about a large system $A$, modeled by a joint state $\rho_{AE}$. In many settings, one wants this information to be mapped to global properties of the system $A$. This allows the information not to be affected by transformations (such as noise) provided they act on a small enough subsystem $B$. Such a condition is described formally by saying that the systems $B$ and $E$ are \emph{decoupled}, i.e., $\rho_{BE} = \rho_B \otimes \rho_E$. In other words, this describes the absence of correlations between $B$ and $E$.  This condition naturally arises in the context of quantum error correcting codes, where information about which state was encoded must be unavailable on any corrupted subsystem, and in the notion of topological order, where information becomes stored in a topological degree of freedom and is inaccessible to measurements on a topologically trivial region.

A decoupling statement generally has the following form: applying a typical unitary transform chosen from some specified set to the system $A$ leads to a state $\rho_{BE} \approx \rho_{B} \otimes \rho_{E}$, provided $B$ is small enough compared to the initial correlations between $A$ and $E$.
A statement of this form is essential in proving a coding theorem for many information processing tasks. But taking the point of view of decoupling for proving coding theorems is especially useful in quantum information, mainly because of the notion of purification. Decoupling appears now as the most successful technique for analyzing quantum information processing tasks. Such an approach was used to study very general quantum information processing tasks like state merging \cite{HOW05, HOW06, Ber09} and fully quantum Slepian-Wolf \cite{ADHW09}, but also in many other settings \cite{Dup09}. For each of these task, a specific decoupling statement was proved but recently Dupuis et al. \cite{DBWR10} proved a very general essentially tight decoupling theorem from which the previously mentioned results can be derived.

The notion of decoupling when $A$ is classical is also studied under the name of privacy amplification. The maps that are applied in order to obtain decoupling are known as randomness extractors, a combinatorial object that is heavily studied in the context of complexity theory and cryptography; see \cite{vadhan:survey} for a survey on this topic. Quantum uncertainty relations can also be viewed as decoupling statements \cite{BFW12}.

Ideas from quantum information related to decoupling have also been used in the context of thermodynamics. For example, del Rio et al. \cite{dRARDV11} used the decoupling theorem of \cite{DBWR10} to study the work cost of an erasure in a fully quantum context. Also, general conditions under which thermal equilibrium is reached are analyzed in \cite{HW13, adrian:thesis, lidia:inprep}. In a different area, Hayden and Preskill \cite{HP07} proved that an $m$-qubit quantum state that was dropped into a black hole could be recovered with high fidelity  from an amount of Hawking radiation containing slightly more than $m$ qubits of quantum information, as long as the dynamics of the black hole approximates a unitary two-design sufficiently well. The speed at which decoupling occurs is particularly important for this question and it motivated the study of fast scramblers \cite{SS08, LSHOH11}.

\subsection{Decoupling with random quantum circuits}
In this paper, we are interested in understanding the dynamics that lead to decoupling. For example, in a system with $n$ particles with only pairwise interactions, how long does it take for the correlations with some observer $E$ to become global? The time required by the dynamics generated by such a Hamiltonian is roughly equivalent to the depth of a corresponding quantum circuit. Thus, in terms of computational complexity, we want to determine what is the minimum size, and particularly, depth for a family of quantum circuits that leads to a decoupled state?

We consider the simple but natural model of random quantum circuits, in which $t$ random gates are applied to randomly chosen pairs of qubits. Random quantum circuits of polynomial size are efficient implementations that are meant to inherit many properties of completely random unitary transformations, which typically require a circuit decomposition which is exponentially large in system size. An important property of interest is that a random unitary maps product states into highly entangled states \cite{HLW04}. As Haar random states are not physical in the sense of computational complexity, it is interesting to determine whether such generic entanglement can be achieved by efficient random quantum circuits.
A lot of work has been done in analyzing convergence properties of the distribution defined by random quantum circuits to the Haar measure on the full unitary group acting on $n$ qubits \cite{RM, ELL05, ODP07, Znidaric2, HL09, Low10, BVPRL, BHH12, TG07, HSZ12} especially properties related to the second moment.
Specifically, Harrow and Low \cite{HL09} proved that random quantum circuits are approximate two-designs with $O(n^2)$ gates. Using the result of \cite{SDTR11}, it follows that such random circuits satisfy a decoupling theorem provided the number of gates is $\Omega(n^2)$. Such a circuit has at least depth $\Omega(n)$, which is much larger than the simple signaling lower bound of $\Omega(\log n)$. 
Another, arguably less natural random circuit model defined in \cite{DCEL09} was shown to decouple a constant size observer $E$ from any macroscopic size subsystem in depth $O(\log n )$. However, it requires a depth proportional to the size of $E$ in general, and thus requires a circuit with depth that is linear in the system size in general. This can be shown using the exact solution to the convergence properties of this model given in \cite{BV13}.

\subsection{Results}
We prove that random quantum circuits with $t = O(n \log^2 n)$ gates achieve essentially optimal decoupling, improving on the results of \cite{HL09} combined with \cite{SDTR11}, which proved this result for $t = O(n^2)$.
Then, by applying gates that act on disjoint qubits in parallel, we show that this circuit runs in time $O(\log^3 n)$.

\subsubsection{Proof technique}

The first step of the proof is to relate the property of interest to the second moment operator of the random quantum circuit. For the random quantum circuits we consider, this moment operator, when evaluated in the Pauli basis, can be seen as the transition matrix of a Markov chain on the Pauli basis elements. The property of decoupling can be formulated in terms of this Markov chain. The convergence times of such Markov chains arising from the second order moments have been previously studied in \cite{ODP07,  Znidaric2, HL09}. However, these convergence times are not sufficient to prove the result we are aiming for and can only give useful bounds when $\Omega(n^2)$ gates are applied. Instead, we analyze the Markov chain in a finer way by bounding the probabilities of going from an initial Pauli string of weight $\ell$ to a Pauli string of weight $k$ within $O(n \log^2 n)$ steps. This is proved by building on the techniques used in \cite{HL09}.

Another reason to see why the methods of \cite{Znidaric2, HL09,   BHH12} cannot lead to the results we obtain here is that they use the spectral gap of the moment operator. It can be shown that this spectral gap only weakly depends on the underlying interaction graph of the circuit \cite{BV13}. But clearly, if the interaction graph is for example a one-dimensional line, Lieb-Robinson bounds give a lower bound of $\Omega(n)$ on the circuit depth at which decoupling can happen. Recalling that our aim is to prove decoupling after a polylogarithmic number of steps, the method we use cannot rely only on the spectral gap of the moment operator.

\subsection{Applications}
Our results show that many information processing tasks in the quantum setting can have very efficient encoding circuits with almost linear size and polylogarithmic depth in the system size. In particular, we can achieve the quantum capacity of the erasure channel to within an arbitrary error using such an encoding circuit. The measurements for optimal quantum state merging can also be implemented using such circuits. Our main technical result can also be used to show that almost-linear sized random quantum circuits define codes with distances that achieve the quantum Gilbert-Varshamov bound; see \cite{BF13isit} for details. To our knowledge capacity achieving codes of such short depth are only known for the quantum polar codes \cite{RDR12, WR12, SRDR13}, which for some special channels can even be efficiently decoded. We note that though inefficient to decode, a code defined by a short depth random quantum circuit is insensitive to which qubits the information to be encoded is initially located.

From a thermodynamics viewpoint, decoupling can be seen as a strong form of thermalization. We refer the reader to recent works that used decoupling theorems in order to derive general conditions under which thermal equilibrium is achieved \cite{HW13, adrian:thesis, lidia:inprep}.  As such, we believe that our results shed light on the speed at which thermal equilibrium is reached for generic two-body dynamics. A simple lower bound for the speed at which global thermal equilibrium can be reached for a closed quantum system is given by Lieb-Robinson bounds on the speed at which a signal can travel under such dynamics.  For pairwise interactions on a complete graph this is given by time $\log n$.  Our results show that this lower bound is almost achieved by a family of time dependent two-body Hamiltonians, specifically those that generate the random quantum circuit model we study.

Whether or not decoupling can be accomplished at the time scale of the signaling speed is relevant to the study of fast scramblers, \cite{SS08}, which was motivated by questions pertaining to quantum information processing in a black hole  \cite{HP07}. Specifically, it was estimated that if the dynamics of the black hole can encode a message dropped into it in a time  $O(\sqrt{n}\log n)$, which is just larger than the lower bound from signaling on the two-dimensional "stretched horizon" of the black hole, then  a violation of the quantum no cloning principle assuming complementarity at the event horizon could occur.  Our results imply that ``infinite''  dimensional random quantum circuits are (pretty) fast scramblers in a strong sense, i.e., scramble a message of any size in $O(\log^3 n)$ time.

\subsection{Organization}
Section \ref{sec:prelim} introduces some basic notation and the model of random quantum circuits we consider here. In Section \ref{sec:decoupling} we state our main result on decoupling with random circuits and reduce the problem to the study of a Markov chain $Q$. This Markov chain $Q$ is studied in Section \ref{sec:mc-analysis}, which is the main technical result of this paper. The fact that the circuits can be parallelized is proved in Section \ref{sec:parallelization}. The appendix contains various technical results that are used in the proofs, such as a generalization of the gambler's ruin lemma and simple estimates for binomial coefficients.

\section{Preliminaries}
\label{sec:prelim}
\subsection{Generalities}
A quantum state for a system $A$ is described by a density operator $\rho \in \cS(A)$ acting on the Hilbert space $A$ associated with the system $A$. A density operator on $A$ lives in the set $\cS(A)$ of positive semidefinite operator with unit trace. If $\rho_{AE}$ describes the joint state on $AE$, the state on the system $A$ is described by the partial trace $\rho_A \eqdef \tr_E \rho_{AE}$. A pure state is a state of rank $1$ and is denoted by $\rho_A = \proj{\rho}_A$ where $\ket{\rho} \in A$. A quantum operation with input system $A$ and output system $C$ is given by a completely positive map $\cT$ that maps operators on $A$ to operators on $C$. A map $\cT$ is said to be completely positive if for any system $B$ and $X \in \cS(A \otimes B)$ we have $(\cT \otimes \id)(X) \geqslant 0$. The system $A$ in this paper is always composed on $n$ qubits, and we denote by $\Phi_{AA'} = \frac{1}{2^n} \sum_{a,a' \in \{0,1\}^n} \ket{a} \bra{a'}_{A} \otimes \ket{a} \bra{a'}_{A'}$ a maximally entangled state between $A$ and $A'$. Here $\{\ket{a}\}$ is the standard basis for $A$. 


Throughout the paper, we use the Pauli basis, which is an orthogonal basis for $2 \times 2$ matrices:
\[
\sigma_0 = \left(
\begin{array}{cc}
1 & 0 \\
0 & 1 \\
\end{array}
\right) \qquad
\sigma_1 = \left(
\begin{array}{cc}
0 & 1 \\
1 & 0 \\
\end{array}
\right)
\qquad
\sigma_2 = \left(
\begin{array}{cc}
0 & -i \\
i & 0 \\
\end{array}
\right)
\qquad
\sigma_3 = \left(
\begin{array}{cc}
1 & 0 \\
0 & -1 \\
\end{array}
\right).
\]
For a string $\nu \in \{0,1,2,3\}^n$, we define $\sigma_{\nu} = \sigma_{\nu_1} \otimes \cdots \otimes \sigma_{\nu_n}$. Observe that $\tr[\sigma_{\nu} \sigma_{\nu'}] = 2^n$ if $\nu = \nu'$ and $0$ otherwise. The support $\supp(\nu)$ of $\nu$ is simply the subset $\{i \in [n] : \nu_i \neq 0\}$ and the weight $|\nu| = |\supp(\nu)|$.
We also need to introduce an entropic quantity to quantify the decoupling accuracy. In particular, for a state $\rho_{AE}$, define 
\begin{equation}
\label{eq:def-h2}
\entHtwo(A|E)_{\rho} = -\log_2\left[\tr\left[\left(\rho_{E}^{-1/4} \rho_{AE} \rho_E^{-1/4}\right)^2\right]\right].
\end{equation}

In order to simplify the statement of the results we use the notation $\poly(n)$ for a number that could be chosen as any polynomial in $n$ and the power of the polynomial can be made large by appropriately choosing the related constants. The set of permutations of $\{1, \dots, n\}$ is denoted by $\mfS_n$.


\subsection{Random quantum circuits}
\label{sec:prelim-rqc}

In a sequential random quantum circuit \textsc{rqc($t$)}, $t$ random two-qubit gates are applied to randomly chosen pairs of qubits sequentially. Here the random two-qubit gate is chosen from the Haar measure on the unitary group acting on two qubits. In fact, our results apply equally well to any gate set whose second-order moment operator is the same as the one for the Haar measure on two qubits. This means that our results would also work if the gates are Clifford unitaries on two qubits. The number of gates of the circuit is one complexity measure but we are also interested in the depth. In this setting, multiple gates can be applied in the same time step as long as they act on disjoint qubits.

We construct a parallelized version of the sequential model in a natural way. Gates are sequentially added to the current level until it is not possible, i.e., there is a gate that shares a qubit with a previously added gate in that level. In this case, a new level is created and the process continues. We then define the parallelized model \textsc{rqc($t$,$d$)} as follows. Choose a random \textsc{rqc($t$)} circuit then parallelize it using the method describe above. If the circuit has depth at most $d$, then we return this circuit, otherwise the circuit is discarded and we restart the procedure.

A model of random circuits of a certain size defines a measure over unitary transformations on $n$ qubits that we call $p_{\circuit}$.
The second-order moment operator will play an important role in all our proofs. The second-order moment operator is a super-operator acting on two copies of the space of operators acting on the ambient Hilbert space, which is an $n$-qubit space in our setting. For a measure $p$ over the unitary group, we can define the second moment operator $\operatorname{M}_p$ as
\[
\operatorname{M}_p[X \otimes Y] = \exc{U \sim p}{UXU^\dagger \otimes UYU^\dagger}.
\]
In particular $\secondmoment_{\haar} = \exc{U \sim p_{\haar}}{UXU^\dagger \otimes UYU^\dagger}$.
Any distribution for which $\secondmoment = \secondmoment_{\haar}$ is referred to as a two-design. We denote by $\Mcirc$ the moment operator for the distribution obtained by applying one step of the random circuit. For the case of  a random unitary distributed according to the Haar measure  applied to a randomly chosen pair $i,j$ of qubits, see e.g., \cite[Section 3.2]{HL09}. We have
\[
\Mcirc = \frac{1}{n(n-1)} \sum_{i \neq j} \operatorname{m}_{ij},
\]
where $\operatorname{m}_{ij}$ only acts on qubits $i$ and $j$ and is defined by
\[
\operatorname{m}_{ij}[\sigma_{\mu} \otimes \sigma_{\mu'}] = \left\{
\begin{array}{ll}
0 & \text{if } \mu \neq \mu' \\
\sigma_{0} \otimes \sigma_{0} & \text{if } \mu = \mu' = 0 \\
\frac{1}{15} \sum\limits_{\nu \in \{0,1,2,3\}^2, \nu \neq 0} \sigma_{\nu} \otimes \sigma_{\nu} &\text{if } \mu = \mu' \neq 0 \\
\end{array}
\right.
\]
for all $\mu, \mu' \in \{0,1,2,3\}^2$. We can thus represent the operator $\Mcirc$ in the Pauli basis using the following $4^n \times 4^n$ matrix
\begin{equation}
\label{eq:markov-chain-q}
Q(\mu, \nu) = \frac{1}{4^n} \tr\left[\sigma_{\nu} \otimes \sigma_{\nu} \Mcirc[\sigma_{\mu} \otimes \sigma_{\mu}] \right].
\end{equation}
In fact, it is simple to verify that $\sum_{\nu \in \{0,1,2,3\}^n} Q(\mu, \nu) = 1$ for all $\mu$ and so $Q$ can be seen as a transition matrix for a Markov chain over the Pauli strings $\{0,1,2,3\}^n$ of length $n$.

Now for a random circuit with $t$ independent random gates applied sequentially, the second moment operator is simply $\Mcirc^t$ and the corresponding matrix in the Pauli basis is also the $t$-th power of $Q$. The properties we are interested in can be expressed as quadratic functions of the entries of the unitary transformation defined by the circuit and thus can be computed from the second moment operator. This means that these properties can be completely reduced to studying the evolution of the Markov chain defined by $Q$.

\section{Decoupling with random quantum circuits}
\label{sec:decoupling}

We start by describing the setting for the general decoupling theorem of \cite{DBWR10}. Consider a state $\rho_{AE}$ on $AE$ and a quantum channel, i.e., a completely positive trace preserving map $\cT : \cS(A) \to \cS(B)$. For example, $\cT$ might be the partial trace map keeping only the qubits in some subsystem $B$. See Figure \ref{fig:decoupling} for an illustration. The theorem gives a sufficient condition on entropic quantities on the state $\rho_{AE}$ and the state $\tau_{A'B} = \cT \otimes \id_{A'}(\Phi_{AA'})$ where $\Phi_{AA'} = \frac{1}{2^n} \sum_{a,a'} \ket{a} \bra{a'}_A \otimes \ket{a} \bra{a'}_{A'}$ is a maximally entangled state on $AA'$. The definition of the entropy $\entHtwo$ is given in \eqref{eq:def-h2}.
\begin{theorem}[General one-shot decoupling \cite{DBWR10}]
\label{thm:decoupling}
With the notation above,
\begin{equation}
\label{eq:decoupling}
\exc{U}{\| \cT(U \rho_{AE} U^{\dagger} ) - \tau_B \otimes \rho_E \|_1} \leq 2^{-\frac{1}{2}\left(\entHtwo(A|E)_{\rho} + \entHtwo(A|B)_{\tau}\right)},
\end{equation}
where $U$ is distributed according to the Haar measure over unitaries acting on $A$.
\end{theorem}
\begin{figure}
\caption{A unitary $U$ (which is going to be a random circuit in this paper) is applied to system $A$ followed by a map $\cT$.}
\label{fig:decoupling}
\begin{center}
\includegraphics{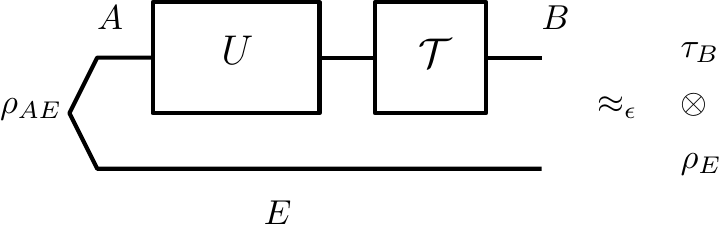}
\end{center}
\end{figure}


In this section, we prove the main result of this paper which is a result analogous to Theorem \ref{thm:decoupling} but where $U$ is a unitary defined by applying a random circuit with $t = O(n \log^2 n)$ gates. Before proving the theorem, we provide a brief overview of the proof. Consider for simplicity that $\cT$ is a partial trace map. We start by relating the trace distance of \eqref{eq:decoupling} to the purity $\tr[\cT(U\tilde{\rho}_{AE}U^{\dagger})^2]$ of the operator $\tilde{\rho}_{AE} = \rho_E^{-1/4} \rho_{AE} \rho_E^{-1/4}$. This step is standard and used in basically all decoupling theorems. Decomposing $\tilde{\rho}_{AE}$ using the Pauli basis on $A$, we can write
\begin{align}
\tilde{\rho}_{AE} = \frac{1}{2^n} \sum_{\nu \in \{0,1,2,3\}^{n} } \sigma_{\nu} \otimes \tr_A[\sigma_{\nu}\tilde{\rho}_{AE}] \quad \text{ and } \quad
\tr[\tilde{\rho}_{AE}^2] = \frac{1}{2^n} \sum_{\nu \in \{0,1,2,3\}^{n} } \tr\left[ \tr_A[\sigma_{\nu}\tilde{\rho}_{AE}]^2 \right].
\label{eq:initial-state}
\end{align}
Note that $\tr[\tilde{\rho}_{AE}^2]$ does not change when a unitary is applied on the system $A$. However, if we apply a unitary $U$ and then keep a subset $S$ of the qubits of $A$, the purity of the reduced state $\tr[ \tr_{S^c}[U\tilde{\rho}_{AE}U^{\dagger}]^2] = \frac{1}{2^{|S|}} \sum_{\nu \in \{0,1,2,3\}^{|S|} } \tr\left[ \tr_A[\sigma_{\nu}U\tilde{\rho}_{AE}U^{\dagger}]^2 \right]$ in general depends on $U$. Observe for example that we only have terms $\tr\left[ \tr_A[\sigma_{\nu}U\tilde{\rho}_{AE}U^{\dagger}]^2 \right]$ where the weight of $\nu$ is at most $|S|$. It then becomes clear that in order to prove that $\tr[ \tr_{S^c}[U\tilde{\rho}_{AE}U^{\dagger}]^2]$ is small when the subsystem $S$ is sufficiently small, we should obtain bounds on $\tr\left[ \tr_A[\sigma_{\nu}U\tilde{\rho}_{AE}U]^2 \right]$ when $\nu$ is small. In particular, if $U$ is a random quantum circuit with $t$ gates, $\ex{\tr\left[ \tr_A[\sigma_{\nu}U\tilde{\rho}_{AE}U]^2 \right]}$ can be written as a function of $Q^t( ., \nu)$ where $Q$ is the transition matrix of the Markov chain introduced in \eqref{eq:markov-chain-q} and using the decomposition of the initial state $\tilde{\rho}_{AE}$. The stationary distribution is given by the uniform distribution over all Pauli strings excluding the identity, $p_Q(\nu) = \frac{1}{4^n-1}$. The main technical result is then to prove that starting at a Pauli string, $\sigma_\mu$ of weight $\ell$, we have that $\sum_\nu |Q^t(\mu, \nu) - p(\nu)| \leq \frac{1}{3^{\ell} \binom{n}{\ell}}$ where $p(\nu) \lesssim p_Q(\nu)$ provided $t > c n \log^2 n$. Note that when computing a mixing time, the worst case over all $\mu$ is considered. Note that the claimed bound on the distance improves with the weight $\ell = |\mu|$. 
For the result we aim to prove, obtaining this explicit dependence on $\ell$ is crucial.
\begin{theorem}
\label{thm:decoupling-seq}
Let $\rho_{AE} \in \cS(AE)$ be an initial arbitrary mixed state and $U_t \rho_{AE} U^{\dagger}_{t}$ be the corresponding state after the application of $t$ random two-qubit gates on the $A$ system, which is composed of $n$ qubits. Let $\cT : \cS(A) \to \cS(B)$ be a completely positive trace preserving map.
Define $\tau_{A'B} = \cT \otimes \id_{A'} (\Phi_{AA'})$, where $\ket{\Phi}_{AA'} = \frac{1}{2^{n/2}}\sum_{a \in \{0,1\}^n} \ket{a}_{A} \ket{a}_{A'}$.

Then we have for any $\delta > 0$, there exists a constant $c$ such that for all $n$ and all $t \geq c n \log^2 n$
\begin{equation}
\label{eq:l1-statement}
\exc{U_t}{ \left\| \cT(U_t\rho_{AE}U_t^{\dagger}) - \tau_B \otimes \rho_{E} \right\|_1} \leq \sqrt{\frac{1}{\poly(n)} + 4^{\delta n} \cdot 2^{-\entHtwo(A|B)_{\tau}} \cdot 2^{-\entHtwo(A|E)_{\rho}}},
\end{equation}
where the expectation is take over the choice of random circuit of size $t$.
\end{theorem}

\begin{proof}
As in \cite{SDTR11}, we use the following H\"older-type inequality for operators $\| \alpha \beta \gamma \|_1 \leq \| |\alpha|^4 \|^{1/4}_1 \| |\beta|^2 \|^{1/2}_1 \| |\gamma|^4 \|^{1/4}_1$, see e.g., \cite[Corollary IV.2.6]{Bha97}.
\begin{align}
\left\| \cT(U_t\rho_{AE}U_t^{\dagger}) - \tau_B \otimes \rho_E \right\|^2_1
&\leq \| (\tau_B^{1/4} \otimes \rho_E^{1/4})^4 \|_1 \cdot \tr \left[ \left( \tau_B^{-1/4} \otimes \rho_E^{-1/4} \left( \cT(\rho_{AE}(t)) - \tau_B \otimes \rho_E \right) \tau_B^{-1/4} \otimes \rho_E^{-1/4} \right)^2 \right] \notag.
\end{align}
Taking the expectation, we have
\begin{align}
\ex{\left\| \cT(U_t\rho_{AE}U_{t}^{\dagger}) - \tau_A \otimes \rho_E \right\|^2_1} &\leq
 \ex{\tr[\tilde{\cT}(U_t \tilde{\rho}_{AE} U_t^{\dagger})^2]} - 2\ex{\tr[ \tilde{\cT}(U_t\tilde{\rho}_{AE}U_t^{\dagger}) \cdot \tilde{\tau}_B \otimes \tilde{\rho}_E]} + \tr[(\tilde{\tau}_B \otimes \tilde{\rho}_E)^2] \notag \\
&\leq \ex{\tr\left[\tilde{\cT}(U_t \tilde{\rho}_{AE} U_t^{\dagger})^2\right]} - \tr[\tilde{\tau}_B^2] \tr[ \tilde{\rho}^2_E ] + \frac{1}{\poly(n)}, \label{eq:l1-to-purity}
\end{align}
where we defined $\tilde{\rho}_{AE} = \rho_E^{-1/4} \rho_{AE} \rho_E^{-1/4}$ and $\tilde{\cT}(.) = \tau_B^{-1/4} \cT(.) \tau_B^{-1/4}$. If the map $\cT$ is such that $\cT(\id)$ is a multiple of the identity then the last line follows directly without using any properties of $U_t$. If this is not the case, we explicitly bound the expectation and obtain the additional $1/\poly(n)$ term, which captures the fact that $\{U_t\}$ form an approximate $1$-design; see Appendix \ref{sec:app-one-design} for a proof of this fact. We also use the fact that $\tr[\tilde{\tau}_B^2] = \tr[\tilde{\rho}_E^2] = 1$. To avoid complicating the expressions, we drop the $1/\poly(n)$ term in the remainder of the proof, as it is taken into account in the final desired statement.


Note that by definition $\tr[\tilde{\rho}^2_{AE}] = 2^{-\entHtwo(A|E)_{\rho}}$. Moreover, since $\Phi_{AA'} = \frac{1}{4^n} \sum_{\nu} \sigma_{\nu} \otimes \sigma_{\nu}$, we have $2^{-\entHtwo(A|B)_{\tau}} = \frac{1}{8^n} \sum_{\nu} \tr[\tilde{\cT}(\sigma_{\nu})^2]$. To compute $\tr[ \tilde{\cT}(U_t\tilde{\rho}_{AE}U_t^{\dagger})^2 ]$, we decompose $U_t\tilde{\rho}_{AE}U_t^{\dagger}$ in the Pauli basis on $A$ as follows:
\begin{align}
\label{eq:decomposition-pauli}
U_t\tilde{\rho}_{AE}U_t^{\dagger} = \frac{1}{2^n} \sum_{\nu \in \{0,1,2,3\}^{n} } \sigma_{\nu} \otimes \tr_A[\sigma_{\nu}U_t\tilde{\rho}_{AE}U_t^{\dagger}].
\end{align}
Applying $\tilde{\cT}$, we get
\begin{align*}
\tilde{\cT}(U_t\tilde{\rho}_{AE}U_t^{\dagger}) &= \frac{1}{2^n}\sum_{\nu \in \{0,1,2,3\}^{n}} \tilde{\cT}(\sigma_{\nu}) \otimes \tr_A[\sigma_{\nu}U_t \tilde{\rho}_{AE}U_t^{\dagger}] \\
&= \frac{1}{4^n}\sum_{\nu, \xi \in \{0,1,2,3\}^{n}} \tr[\sigma_{\xi} \tilde{\cT}(\sigma_{\nu})] \sigma_{\xi} \otimes \tr_A[\sigma_{\nu}U_t \tilde{\rho}_{AE} U_t^{\dagger}].
\end{align*}
As a result, we have
\begin{align*}
\tr [\tilde{\cT}(U_t\tilde{\rho}_{AE}U_t^{\dagger})^2] &= \frac{1}{2^n} \sum_{\xi \in \{0,1,2,3\}^n} \tr\left[\left(\frac{1}{2^n} \sum_{\nu} \tr[ \sigma_{\xi} \tilde{\cT}(\sigma_{\nu}) ] \tr_A[\sigma_{\nu} U_t \tilde{\rho}_{AE} U_t^{\dagger} ]\right)^2 \right] \\
			&= \frac{1}{2^n} \sum_{\xi \in \{0,1,2,3\}^n} \frac{1}{4^n} \tr[\sigma_{\xi} \tilde{\cT}(\id_A)]^2 \tr[\tilde{\rho}^2_E] \\
			&+ \frac{1}{8^n} \sum_{\xi, \nu, \nu' \in \{0,1,2,3\}^n, \nu \text{ or } \nu' \neq 0} \tr[\sigma_{\xi} \tilde{\cT}(\sigma_{\nu})] \tr[\sigma_{\xi} \tilde{\cT}(\sigma_{\nu'})] \cdot \tr\left[ \tr_A[\sigma_{\nu} U_t\tilde{\rho}_{AE}U_t^{\dagger} ]  \tr_A[\sigma_{\nu'} U_t\tilde{\rho}_{AE}U_{t}^{\dagger} ] \right] \\
			&= \tr[\tilde{\tau}_B^2] \tr[\tilde{\rho}^2_E] + \frac{1}{8^n} \sum_{\nu, \nu' \in \{0,1,2,3\}^n, \nu \text{ or } \nu' \neq 0} T_{\nu, \nu'} \cdot \tr\left[ \tr_A[\sigma_{\nu} U_t \tilde{\rho}_{AE}U_t^{\dagger}] \tr_A[\sigma_{\nu'} U_t\tilde{\rho}_{AE}U_t^{\dagger} ] \right],
\end{align*}
where we defined $T_{\nu, \nu'} = \sum_{\xi} \tr[\sigma_{\xi} \tilde{\cT}(\sigma_{\nu})] \tr[\sigma_{\xi} \tilde{\cT}(\sigma_{\nu'})]$.
Getting back to equation \eqref{eq:l1-to-purity} and using the concavity of the square root function, we have
\begin{align}
\ex{\left\| \cT(U_t\rho_{AE}U_t^{\dagger}) - \tau_B \otimes \rho_E \right\|_1} \leq \sqrt{ \ex{
 \frac{1}{8^n} \sum_{\nu, \nu' \in \{0,1,2,3\}^n, \nu \text{ or } \nu' \neq 0} T_{\nu, \nu'} \cdot \tr\left[ \tr_A[\sigma_{\nu} U_t \tilde{\rho}_{AE} U_t^{\dagger} ] \right] \tr\left[ \tr_A[\sigma_{\nu'} U_t\tilde{\rho}_{AE}U_t^{\dagger} ] \right]
} }.
\label{eq:bound-l1-general}
\end{align}
Observe that this term is a quadratic function of $U_t$ and thus only depends on the second moment operator $\mathrm{M}$ of our distribution over unitary transformations on $A$. Recall that the second moment operator is a super operator acting on operators acting on two copies of $A$. For a random quantum circuit with $t$ gates, the second moment operator is $\mathrm{M}_{\circuit}^t$. We have for any $\nu, \nu'$,
\begin{align}
\ex{ \tr\left[ \tr_A[\sigma_{\nu} U_t \tilde{\rho}_{AE} U_t^{\dagger} ] \tr_A[\sigma_{\nu'} U_t \tilde{\rho}_{AE} U_t^{\dagger} ] \right] } &= \ex{\tr\left[ \tr_{A} [ \sigma_{\nu} U_t \tilde{\rho}_{AE} U_{t}^{\dagger} ] \otimes  \tr_{A'}[\sigma_{\nu'} U_t \tilde{\rho}_{A'E'} U_t^{\dagger} ]  F_{EE'} \right]} \notag \\
&= \tr\left[ \tr_{AA'}[ \sigma_{\nu} \otimes \sigma_{\nu'} (\mathrm{M}^t_{\circuit} \otimes \id_{EE'})[ \tilde{\rho}_{AE} \otimes \tilde{\rho}_{A'E'}] F_{EE'}\right], \label{eq:introduce-m}
\end{align}
where we used in the first equality the fact that $\tr[\omega_E \omega'_E] = \tr[\omega_E \otimes \omega'_{E'} F_{EE'}]$ with $F_{EE'}$ being the swap operator. By expanding the initial state $\tilde{\rho}_{AE}$ in the Pauli basis, we obtain
\begin{align*}
(\mathrm{M}^t_{\circuit} \otimes \id_{EE'})[\tilde{\rho}_{AE} \otimes \tilde{\rho}_{A'E'}]
&= \frac{1}{4^n} \sum_{\mu, \mu' \in \{0,1,2,3\}^n}(\mathrm{M}^t_{\circuit} \otimes \id_{EE'})\left[ \sigma_{\mu} \otimes \tr_{A}[ \sigma_\mu \tilde{\rho}_{AE} ] \otimes \sigma_{\mu'} \otimes \tr_{A'}[ \sigma_{\mu'} \tilde{\rho}_{A'E'} ] \right]\\
&= \frac{1}{4^n} \sum_{\mu, \mu' \in \{0,1,2,3\}^n} \mathrm{M}^t_{\circuit}[\sigma_{\mu} \otimes \sigma_{\mu'}] \otimes  \tr_{A}[ \sigma_\mu \tilde{\rho}_{AE} ] \otimes \tr_{A'}[ \sigma_{\mu'} \tilde{\rho}_{A'E'} ].
\end{align*}
Continuing, we get
\begin{align*}
\ex{\tr\left[ \tr_A[\sigma_{\nu} \tilde{\rho}_{AE}(t) ] \tr_A[\sigma_{\nu'} \tilde{\rho}_{AE}(t) ] \right]}
&= \frac{1}{4^n} \sum_{\mu, \mu' \in \{0,1,2,3\}^n} \tr\left[\sigma_{\nu} \otimes \sigma_{\nu'} \mathrm{M}^t_{\circuit}[\sigma_{\mu} \otimes \sigma_{\mu'}] \right] \otimes  \tr\left[\tr_{A}[ \sigma_\mu \tilde{\rho}_{AE} ] \tr_{A}[ \sigma_{\mu'} \tilde{\rho}_{AE} ] \right].
\end{align*}
Recall that $\frac{1}{4^n}\tr\left[\sigma_{\nu} \otimes \sigma_{\nu'} \mathrm{M}^t_{\circuit}[\sigma_{\mu} \otimes \sigma_{\mu'}] \right] = Q^t(\mu, \nu)$ if $\mu' = \mu$ and $\nu = \nu'$ and $0$ otherwise. The expectation in equation \eqref{eq:bound-l1-general} then becomes
\begin{align}
&\frac{1}{8^n} \sum_{\nu \in \{0,1,2,3\}^n, \nu \neq 0} T_{\nu, \nu} \sum_{\mu \in \{0,1,2,3\}^n} Q^t(\mu, \nu) \tr[\tr_{A}[ \sigma_\mu \tilde{\rho}_{AE} ]^2 ] \notag \\
&= \frac{1}{4^n} \sum_{\nu \in \{0,1,2,3\}^n, \nu \neq 0} \tr[ \tilde{\cT}(\sigma_{\nu})^2] \sum_{\mu \in \{0,1,2,3\}^n, \mu \neq 0} Q^t(\mu, \nu) \tr[\tr_{A}[ \sigma_\mu \tilde{\rho}_{AE} ]^2 ] \notag \\
&= \frac{1}{4^n} \sum_{\mu \in \{0,1,2,3\}^n, \mu \neq 0}  \tr[\tr_{A}[ \sigma_\mu \tilde{\rho}_{AE} ]^2 ] \sum_{\nu \in \{0,1,2,3\}^n, \nu \neq 0} \tr[ \tilde{\cT}(\sigma_{\nu})^2] Q^t(\mu, \nu).
\label{eq:string-to-weight}
\end{align}

The main technical result in this proof is in Theorem \ref{thm:mc-Qconvergence} (which we defer to Section \ref{sec:mc-analysis}), where we obtain a bound of
\begin{equation}
\label{eq:prob-bound-ell-k}
\sum_{\nu \in \{0,1,2,3\}^n, \nu \neq 0} \left| Q^t(\mu, \nu) - p_\delta(\nu) \right| \leq \frac{1}{(3-\eta)^{\ell} \binom{n}{\ell} \poly(n)},
\end{equation}
where $p_{\delta}(\nu) \le \frac{4^{\delta n}}{4^n-1}$ and $|\mu|=\ell$ and for any positive constants $\delta$ and $\eta$ and $t \geq c n \log^2 n$ for some constant $c$ depending on $\delta$ and $\eta$ and the desired polynomial. We have by plugging equation \eqref{eq:prob-bound-ell-k} into \eqref{eq:string-to-weight}, we obtain
\begin{align}
& \ex{\frac{1}{8^n} \sum_{\nu, \nu' \in \{0,1,2,3\}^n, \nu \text{ or } \nu' \neq 0} T_{\nu, \nu'} \cdot \tr\left[ \tr_A[\sigma_{\nu} U_t \tilde{\rho}_{AE}U_t^{\dagger} ] \right] \tr\left[ \tr_A[\sigma_{\nu'} U_t \tilde{\rho}_{AE} U_t^{\dagger} ] \right]} \notag \\
&= \frac{1}{4^n} \sum_{\ell = 1}^n \sum_{\mu: |\mu| = \ell} \tr[\tr_{A}[ \sigma_\mu \tilde{\rho}_{AE} ]^2 ] \sum_{\nu \in \{0,1,2,3\}^n, \nu \neq 0} \tr[ \tilde{\cT}(\sigma_{\nu})^2] \left( p_{\delta}(\nu) + Q^t(\mu, \nu) - p_{\delta}(\nu) \right) \notag \\
&\leq \frac{1}{4^n} \sum_{\mu \neq 0} \tr[\tr_{A}[ \sigma_\mu \tilde{\rho}_{AE} ]^2 ] \sum_{\nu \neq 0} \tr[ \tilde{\cT}(\sigma_{\nu})^2] \frac{4^{\delta n}}{4^n-1} \notag \\
&+ \frac{1}{4^n} \sum_{\ell = 1}^n \sum_{\mu: |\mu| = \ell} \tr[\tr_{A}[ \sigma_\mu \tilde{\rho}_{AE} ]^2 ]  \frac{1}{(3-\eta)^{\ell} \binom{n}{\ell} \poly(n)} \max_{\nu} \tr[\tilde{\cT}(\sigma_{\nu})^2]. \label{eq:prob-ell-k}
\end{align}
Let us start by considering the first term. Recall that $\sum_{\mu} \tr[\tr_A[\sigma_{\mu} \tilde{\rho}_{AE}]^2] = 2^n \tr[\tilde{\rho}^2_{AE}]$ and $\frac{1}{8^n} \sum_{\nu} \tr[\tilde{\cT}(\sigma_{\nu})^2] = 2^{-\entHtwo(A|B)_{\tau}}$. As a result,
\begin{align*}
\frac{1}{4^n} \sum_{\mu \neq 0} \tr[\tr_{A}[ \sigma_\mu \tilde{\rho}_{AE} ]^2 ] \sum_{\nu \neq 0} \tr[ \tilde{\cT}(\sigma_{\nu})^2] \frac{4^{\delta n}}{4^n-1}
&= 4^{\delta n} \frac{1}{4^n} \sum_{\nu \neq 0} \tr[ \tilde{\cT}(\sigma_{\nu})^2] \frac{2^n \tr[\tilde{\rho}_{AE}^2] - \tr[\tilde{\rho}_E^2]}{4^n - 1} \\
&\leq 4^{\delta n} \frac{1}{8^n} \sum_{\nu} \tr[ \tilde{\cT}(\sigma_{\nu})^2] \frac{2^n\tr[\tilde{\rho}_{AE}^2] - \tr[\tilde{\rho}_E^2]}{2^n - 1} \\
&\leq 4^{\delta n} 2^{-\entHtwo(A|B)_{\tau}} 2^{-\entHtwo(A|E)_{\rho}}.
\end{align*}
To prove that the second term can be bounded by an inverse polynomial, we use Lemma \ref{lem:level} which is proven in the appendix. It states that
\begin{equation}
\label{eq:level}
\sum_{\nu : |\nu| = \ell} \tr\left[ \tr_A[ \sigma_{\nu} \tilde{\rho}_{AE}]^2 \right] \leq 12n^4 \cdot (3-\eta)^{\ell} \binom{n}{\ell}
\end{equation}
provided $\tr[\tilde{\rho}^2_{AE}] \leq 2^{(1-\delta)n}$. Also, we have for any $\nu \in \{0,1,2,3\}^n$,
\begin{align*}
\tr[ \tilde{\cT}(\sigma_{\nu})^2 ] &= \tr\left[ \cT(\id/2^n)^{-1/2} \cT(\sigma_{\nu}) \cT(\id/2^n)^{-1/2} \cT(\sigma_{\nu}) \right] \\
&\leq \tr[\id \sqrt{2^n} \sigma_{\nu} \id \sqrt{2^n} \sigma_{\nu}] \\
&= 4^n,
\end{align*}
using the monotonicity of the relative entropy of order 2; see e.g., \cite{DFW13}.
Plugging the value of $\eta$ from \eqref{eq:level} into the second term of \eqref{eq:prob-ell-k}, we obtain
\begin{align*}
\frac{1}{4^n} \sum_{\ell = 1}^n \sum_{\mu: |\mu| = \ell} \tr[\tr_{A}[ \sigma_\mu \tilde{\rho}_{AE} ]^2 ]  \frac{1}{(3-\eta)^{\ell} \binom{n}{\ell} \poly(n)} \max_{\nu} \tr[\tilde{\cT}(\sigma_{\nu})^2]
&\leq \frac{1}{4^n} \max_{\nu} \tr[\tilde{\cT}(\sigma_{\nu})^2] \cdot \frac{12 n^5}{\poly(n)} \\
&\leq \frac{1}{\poly(n)},
\end{align*}
by choosing a large enough $c$. Note that in the case where $\tr[\tilde{\rho}^2_{AE}] > 2^{(1-\delta)n}$, the theorem clearly holds because the upper bound is greater than $2$.
\end{proof}

An important example for the map $\cT$ is the partial trace map.
\begin{corollary}
\label{cor:decoupling-partial-trace}
Let $\rho_{AE}$ be an initial arbitrary mixed state on $n$ qubits and $U_t \rho_{AE} U^{\dagger}_{t}$ be the corresponding state after the application of $t$ random two-qubit gates on the $A$ system. Then let $S$ be a random subset of the qubits $\{1, \dots, n\}$ of size $s$.
%

Then we have for any constant $\delta > 0$, there exists a constant $c$ such that $t \geq c n \log^2 n$, we have for subset $S$ of size $s$:
\begin{equation}
\label{eq:l1-statement-pt}
\exc{U_t}{ \left\| \tr_{A_{S^c}}\left[U_t\rho_{AE}U_t^{\dagger}\right] - \frac{\id_{A_S}}{2^s} \otimes \rho_{E} \right\|_1} \leq \sqrt{\frac{1}{\poly(n)} + 4^{\delta n} \cdot 2^{2s-n} \cdot 2^{-\entHtwo(A|E)_{\rho}}}.
\end{equation}

\end{corollary}
\begin{proof}
It suffices to compute the entropic quantity for $\cT$. If $\cT$ is the erasure map for all but $s$ qubits, we have
$2^{-\entHtwo(A|B)_{\tau}} = 2^{2s-n}$.
\end{proof}


\subsection{Depth}
\label{sec:parallelization}
We proved in the last section that decoupling can be accomplished using $O(n \log^2 n)$ gates. In this section, we study another complexity measure which is closely related to time: the depth. Gates acting on disjoint qubits are allowed to be executed in parallel. The depth of a circuit with $t$ gates is at most $t$ but it could be much smaller than $t$. In particular, for a random quantum circuit we expect many gates to act on disjoint qubits so that they can be implemented in a number of time steps that can be much smaller than $t$. As mentioned in the preliminaries, to construct the parallelized circuit, one keeps adding gates to the current level until there is a gate that shares a qubit with a previously added gate in that level. In this case, a new level is created and the process continues. In the following proposition, we prove that by parallelizing a random circuit on $n$ qubits having $t$ gates we obtain with high probability a circuit of depth $O(\frac{t}{n} \log n)$.
For the purpose of parallelization, the gates can simply be labelled by the two qubits the gate acts upon.

\begin{proposition}
\label{prop:parallelization}
Consider a random sequential circuit composed of $t$ gates where $t$ is a polynomial in $n$. Then parallelize the circuit as described above. Except with probability $\frac{1}{\poly(n)}$, the resulting circuit has depth at most $O\left(\frac{t}{n} \log n\right)$.
In other words, in the model \textsc{rqc}($c n \log^2 n$, $c' \log^3 n$), discarding a circuit only happens with probability $\frac{1}{\poly(n)}$ provided the constants $c$ and $c'$ are appropriately chosen.
\end{proposition}
In order to prove this lemma, we use the following calculation:
\begin{lemma}
Let $G_1, \dots, G_k$ be a sequence of independent and random gates $G_i \in \binom{n}{2}$, then the probability that $G_1, \dots, G_k$ form a circuit of depth $k$ is at most $\left(\frac{2}{n} \right)^{k-1} \cdot k !$
\end{lemma}
\begin{proof}
We prove this by induction on $k$. For $k = 2$, we may assume $G_1 = (1,2)$, in which case $\pr{G_2 \cap \{1,2\} \neq \emptyset} \leq 4/n$.
Now the probability that $G_1, \dots, G_{k+1}$ form a circuit of depth $k+1$ can be bounded by
\[
\pr{G_1, \dots, G_k \text{ form a circuit of depth $k$ } } \cdot \pr{G_{k+1} \cap \left( G_1 \cup \cdot \cup G_k \right) \neq \emptyset | G_1, \dots, G_k \text{ form a circuit of depth $k$ } }.
\]
Now it suffices to see that, conditioned on $\event{G_1, \dots, G_k \text{ form a circuit of depth $k$}}$, the number of nodes occupied by $G_1, \dots, G_k$ is at most $k+1$. Thus, using this fact and the induction hypothesis, we obtain
a bound of
\[
\left(\frac{2}{n}\right)^{k-1} k! \cdot 2 \cdot \frac{k+1}{n} = \left( \frac{2}{n} \right)^{k} (k+1)! \ ,
\]
which conclude the proof.
\end{proof}

\begin{proof}[of Proposition \ref{prop:parallelization}]
Suppose we apply $m$ gates for some $m$ to be chosen later.
\begin{align*}
\pr{G_1, \dots, G_m \text{ form a circuit of depth at least $d$ } }
&= \pr{ \exists (i_1, \dots, i_d) \in [m]^d : G_{i_1}, \cdots, G_{i_d} \text{ form a circuit of depth d} } \\
&\leq \binom{m}{d}  \left(\frac{2}{n} \right)^{d-1} \cdot d ! \\
&\leq m^d \cdot \left(\frac{2}{n} \right)^{d-1}.
\end{align*}
Now we can fix $m = n/4$ and $d = c \log n + 1$ for some constant $c$ to be chosen depending on the desired probability bound, then we have
\begin{align*}
\pr{G_1, \dots, G_m \text{ form a circuit of depth at least $d$ } }
&\leq m \cdot \left(\frac{2m}{n}\right)^{d-1} \leq n^{-c+1}.
\end{align*}
This proves that every set of $n/4$ gates generates a circuit of depth at most $c\log n + 1$ with probability at least $1-1/n^{-c+1}$, and so if we have $4t/n$ such sets, we get depth at most $4t/n(c \log n + 1)$ with probability at least $1-4t/n^{c}$.
\end{proof}

The next corollary follows directly from Theorem \ref{thm:decoupling-seq} and Proposition \ref{prop:parallelization}.
\begin{corollary}
In the setting of Theorem \ref{thm:decoupling-seq} and if $U_t$ is the unitary computed by a random quantum circuit chosen according to the model \textsc{rqc($cn \log^2 n$,$c' \log^3n$)}, then
\begin{equation}
\label{eq:decoupling-par}
\ex{ \left\| \cT(U_t\rho_{AE}U_t^{\dagger}) - \tau_B \otimes \rho_{E} \right\|_1} \leq \sqrt{\frac{1}{\poly(n)} + 4^{\delta n} \cdot 2^{-\entHtwo(A|B)_{\tau}} \cdot 2^{-\entHtwo(A|E)_{\rho}}}.
\end{equation}
\end{corollary}
\begin{proof}
We write $\depth(U_t)$ for the depth of the circuit obtained by parallelizing the circuit defining $U_t$. Let $t = cn \log^2 n$ and $d = c' \log^3n$. We have
\begin{align*}
\exc{\text{\textsc{rqc}}(t,d)}{\left\| \cT(U_t\rho_{AE}U_t^{\dagger}) - \tau_B \otimes \rho_{E} \right\|_1} &= \exc{\text{\textsc{rqc}}(t)}{\1_{\depth(U_t) \leq d}\left\| \cT(U_t\rho_{AE}U_t^{\dagger}) - \tau_B \otimes \rho_{E} \right\|_1} \\
&\leq \exc{\text{\textsc{rqc}}(t)}{\left\| \cT(U_t\rho_{AE}U_t^{\dagger}) - \tau_B \otimes \rho_{E} \right\|_1} + \ex{(1-\1_{\depth(U_t) \leq c' \log^3 n})} \\
&\leq \exc{\text{\textsc{rqc}}(t)}{\left\| \cT(U_t\rho_{AE}U_t^{\dagger}) - \tau_B \otimes \rho_{E} \right\|_1} + \frac{1}{\poly(n)}.
\end{align*}
\end{proof}

%
%
%
%
%
%

\section{Analysis of the random walk over Pauli operators}
\label{sec:mc-analysis}

This section is devoted to the analysis of the Markov chain $Q$ over strings $\{0,1,2,3\}^n$ introduced in \eqref{eq:markov-chain-q}. The property we study is similar to the mixing time but differing in two ways. First, instead of considering the distance between  the distribution $Q^t(\mu, .)$ obtained after $t$ steps of the Markov chain and the stationary distribution $p_Q$, we can replace $p_Q$ by any distribution that has the property $p \leq 2^{\delta n} p_Q$. In other words, we can compute the distance to any distribution $p$ that has a small max-entropy relative to $p_Q$, i.e., $\textrm{D}_{\max}(p, p_Q) \leq \delta n$. 
Second, the bound we obtain on the distance depends on the initial state $\mu$.

\begin{theorem}
\label{thm:mc-Qconvergence}
Let $Q$ be the Markov chain over Pauli strings defined in \eqref{eq:markov-chain-q}. 
For any constants $\delta \in (0,1/16), \eta \in (0,1)$, there exists a constant $c$ such that for $t \geq c n \log^2 n$,  and all Pauli strings $\sigma_\mu$ of weight $\ell$, and large enough $n$,  there exists a possible subnormalized distribution $p_{\delta}$ on strings $\{0,1,2,3\}^n$ such that for all $\nu$, 
\[
p_{\delta}(\nu) \leq \frac{16^{\delta n}}{4^n-1}
\] 
and 
\[
\sum_{\nu \in \{0,1,2,3\}^n, \nu \neq 0} \left| Q^t(\mu,\nu)- p_\delta(\nu) \right|  \le  \frac{1}{(3-\eta)^\ell \binom{n}{\ell}} \frac{1}{\poly(n)}.\]

\end{theorem}

We first prove a similar result for a  Markov chain which acts only on the weights of the Pauli strings. More precisely, we define $P(\ell, k) = \sum_{\nu: |\nu| = k} Q(\mu, \nu)$ where $\mu$ is an arbitrary string with weight $\ell$. Note that this definition is independent of the choice of $\mu$. This follows from the fact that $Q(\pi(\mu), \pi(\nu)) = Q(\mu, \nu)$ for any permutation $\pi \in \mfS_n$ of the $n$ qubits, and also $Q(\gamma(\mu), \gamma(\nu)) = Q(\mu, \nu)$, where $\gamma \in \mfS_3^{\times n}$ is a relabeling of the operators $\{1,2,3\}$. We have
\begin{equation}
\label{eq:def-P}
P(\ell,k) = \left\{
\begin{array}{cc}
1- \frac{2\ell(3n-2\ell-1)}{5n(n-1)} & \text{ if } k=\ell \\
\frac{2\ell(\ell-1)}{5n(n-1)} & \text{ if } k = \ell - 1\\
\frac{6\ell(n-\ell)}{5n(n-1)} & \text{ if } k = \ell + 1\\
0 & \text{ otherwise.}
\end{array} \right.
\end{equation}
We refer the reader to \cite{HL09} for more details on how to derive the parameters of this Markov chain. In fact, \cite{HL09} study the mixing time of this Markov chain. Here, we need to analyze a slightly different property: starting at some point $\ell$, what is the probability that after $t$ steps the random walk ends up in a point $k$? One can obtain bounds on this probability using the mixing time but these bounds only give something useful for our setting if $t = \Omega(n^2)$. So we need to improve the analysis of \cite{HL09} and compute the desired probability directly.

\begin{theorem}
\label{thm:mc-Pconvergence}
Let $P$ be the Markov chain transition matrix defined in \eqref{eq:def-P}. For any constants $\delta \in (0,1/16), \eta \in (0,1)$, there exists a constant $c$ such that for $t \geq c n \log^2 n$ and all integers  $1 \leq \ell \leq n$ and $1 \leq k \leq n$, we have for large enough $n$
\[
P^t(\ell, k) \leq 4^{\delta n} \cdot \frac{\binom{n}{k} 3^k}{4^n - 1} + \frac{1}{(3-\eta)^{\ell} \binom{n}{\ell}}\frac{1}{\poly(n)}.
\]
\end{theorem}
\begin{proof}
It is convenient for the proof to define variables  $X_0, X_1, \dots, X_t, \dots $ for the Markov chain with transition probabilities $P$.  We write $X_t(\ell)$ for a chain with $X_0 = \ell$. With this notation, $P^t(\ell, k) = \pr{X_t(\ell) = k}$. The stationary distribution of $P$ is given by $\pi(k) = \frac{3^k \binom{n}{k} }{4^n - 1}$ (see \cite[Lemma 5.3]{HL09}). As a result, we have for any $t \geq 1$,
\begin{align}
\frac{1}{4^n-1} \sum_{\ell=1}^n 3^{\ell} \binom{n}{\ell} \pr{X_t(\ell) = k}
&= \frac{3^{k} \binom{n}{k}}{4^n-1}.
\label{eq:stat-distribution}
\end{align}

The general strategy of the proof is as follows.
First we choose two reference points $r_{-}$ and $r_{+}$ with $r_{-} \leq 3n/4 \leq r_+$. The states $r_{-}$ and $r_+$ are chosen for two properties: they should have a significant probability in the stationary distribution of $P$ and moreover they should be bounded away from $3n/4$ so that the probability of reaching $r_-$ starting below can be bounded and similarly for the probability of reaching $r_+$ starting above it. This divides the state space of the chain into three parts: $[1, r_-)$, $[r_-, r_+]$ and $(r_+, n]$. When $\ell \in [r_-, r_+]$, it is simple to prove the desired result. Whenever the starting point of the chain $\ell \in [1, r_-)$ or $\ell \in (r_+, n]$, we prove that the interval $[r_-, r_+]$ is reached with high probability (that depends on $\ell$) if the chain is run for sufficiently long. We then conclude by using the first case. We note that most of the difficulty is in handling the case $\ell \in [1, r_-)$.

We start by picking specifically $r_-$ and $r_+$. We choose $r_{-} = (3/4-\delta)n$ and $r_+ = (3/4+\delta) n$. They satisfy the following properties. The first one is
\begin{align}
\binom{n}{r_-} 3^{r_-} \geq 4^{(1-\delta)n} \qquad \text{ and } \qquad \binom{n}{r_+} 3^{r_+} \geq 4^{(1-\delta)n},
\label{eq:r-mass}
\end{align}
for sufficiently large $n$. To see the second inequality, write
\begin{align*}
\binom{n}{r_+} = \frac{n (n-1) \cdots (n/4+1) \cdot n/4 \cdots \left((1/4-\delta) n + 1\right)}{(3/4n)! \cdot (3/4n + 1) \cdots (3/4+\delta)n} \geq \binom{n}{3n/4} \left(\frac{1-4\delta}{3}\right)^{\delta n}.
\end{align*}
The second property is that for all $x < r_-$ and $y > r_{+}$,
\begin{align}
\frac{P(x, x+1)}{P(x, x-1)} = 3\cdot \frac{n-x}{x-1} \geq 1+2\delta \qquad \text{ and } \qquad \frac{P(y, y-1)}{P(y, y+1)} \geq 1+2\delta.
\label{eq:r-drift}
\end{align}

We now start with the case $\ell \in [r_-, r_+]$. For this, we simply use \eqref{eq:stat-distribution}. For any $t \geq 1$ and any $r \in [r_{-}, r_+]$,
\begin{align}
\pr{X_t(r) = k} &= \frac{4^n-1}{\binom{n}{r} 3^{r}} \cdot \frac{\binom{n}{r} 3^{r}}{4^n-1} \pr{X_t(r) = k} \notag \\
&\leq \frac{4^n-1}{\binom{n}{r} 3^{r}} \cdot
\frac{1}{4^n-1} \sum_{\ell=1}^n 3^{\ell} \binom{n}{\ell} \pr{X_t(\ell) = k} \notag \\
&\leq \frac{4^n-1}{\binom{n}{r} 3^{r}} \cdot \frac{\binom{n}{k} 3^k}{4^n - 1} \notag \\
&\leq 4^{\delta n} \cdot \frac{\binom{n}{k} 3^k}{4^n - 1}. \label{eq:returnfn-from-middle}
\end{align}
In the last line, we used the inequalities in \eqref{eq:r-mass}. This proves the case $\ell \in [r_-, r_+]$, and in fact for any $t \geq 1$.


We now handle the case $\ell \in [1, r_{-})$. Introduce $T_{r_-}(\ell) = \min \{t \geq 1 : X_t(\ell) \geq r_{-} \}$. Note that we have for any $t$
\begin{align}
\pr{X_t(\ell) = k} &\leq \pr{T_{r_-}(\ell) < t, X_t(\ell) = k} + \pr{T_{r_-}(\ell) \geq t} \notag \\
				&= \pr{T_{r_-}(\ell) < t, X_{t-T}(r) = k} + \pr{T_{r_-}(\ell) \geq t} \notag \\
			&\leq  \max_{1 \leq s \leq t} \pr{X_s(r) = k} + \pr{T_{r_-}(\ell) \geq t}. \label{eq:event-ep}
\end{align}
Using \eqref{eq:returnfn-from-middle}, we can bound the first term. The objective of the remainder of the proof is to bound the probability $\pr{T_{r_-} \geq t}$ when $t = c n \log^2 n$. This is done in Lemma \ref{lem:time-reach-middle} below and it concludes the case $\ell \in [1, r_-)$.

The case $\ell \in (r_+, n]$ is analogous, except that we use Lemma \ref{lem:time-reach-middle-+} instead, which has a similar proof but is significantly simpler. We note that in this case, it is possible to obtain a better probability bound without the dependence on the starting point $\ell$.
\end{proof}

\begin{lemma}
\label{lem:time-reach-middle}
Let $\delta \in (0,1/16)$ and $\eta \in (0,1)$ be constants and $r_{-}$ satisfying condition \eqref{eq:r-drift}. Then for a large enough constant $c$ (depending on $\delta$ and $\eta$) and large enough $n$, we have for all $\ell \leq r_-$,
\[
\pr{T_{r_-}(\ell) > c n \log^2 n} \leq 2^{-2n} + \frac{1}{\left(3-\eta\right)^{\ell} \binom{n}{\ell} } \cdot \frac{1}{\poly(n)}.
\]
\end{lemma}
\begin{proof}
As this proof does not involve $r_+$, we write $r$ instead of $r_-$ to make the notation lighter.
To prove this result, we start by defining an accelerated walk $\{Y_i\}$ as in \cite{HL09} and the corresponding stopping time $S = \min \{s : Y_s \geq r\}$. More formally, let $N_0 = 0$ and $N_{i+1} = \min \{ k \geq N_i : X_k \neq X_{N_i}\}$ and then $Y_{i} = X_{N_i}$. It is not hard to see that $\{Y_i\}$ is a Markov chain and the transition probabilities are given by the transition probabilities for $\{X_k\}$ conditioned on moving.

We also define the waiting time $W_i = N_{i+1} - N_{i} - 1$ to be the number of steps it takes the walk to change states. Conditioned on $Y_i$, $W_i$ has a geometric distribution with parameter $\frac{2Y_i(3n-2Y_i-1)}{5 n (n-1)}$. Notice that this distribution is stochastically dominated by a geometric distribution with parameter $\frac{2Y_i}{5n}$, which we sometimes use instead (we are only interested in upper bounds on the waiting times).

Getting back to $T_{r}$ denoted simply $T$ in the following, notice that $T = S+W_1 + W_2 + \dots + W_S$. So we have for all $s$
\begin{equation}
\label{eq:t-s}
\pr{T > t + s} \leq \pr{S > s} + \pr{S \leq s, W_1 + \dots + W_S > t}.
\end{equation}
We will choose $s$ later so that both terms are small. We start by bounding the first term, which can be done using a simple application of a Chernoff-type bound.
\begin{lemma}
\label{lem:bound-prob-s}
If $s > \frac{n}{3\delta}$, we have
\[
\pr{S > s} \leq \exp{-\frac{\delta^2}{18} \cdot s}.
\]
\end{lemma}
\begin{proof}
For this we just use a concentration bound on the position of a random walk relative to its expectation. Recall that the probability of moving forward when $Y_i = r$ is $\frac{6 r (n-r)}{6 r (n-r) + 2 r (r +1)}$. Then, using the property \eqref{eq:r-drift} the probability of moving forward is at most $1/2+\delta/3$ for $Y_i$ provided $Y_i \leq r$. Define a random walk $Y'_i$ with $Y'_0 = 0$ and it moves to the right with probability $1/2 + \delta/3$ and to the left with probability $1/2 - \delta/3$.  For $i \leq S$, we can assume that $Y'_i \leq Y_i$. In other words, we have $S' \geq S$ where $S' = \min \{i : Y'_i \geq r\}$.
Thus,
\begin{align*}
\pr{S > s} &\leq \pr{S' > s} \\
			&\leq \pr{Y'_{s} < r} \\
			&= \pr{Y'_{s} - \ell < 2 \cdot \delta/3 \cdot s - (2\delta/3 \cdot s+\ell-r)} \\
			&\leq \exp{-\frac{(2\delta/3 s+\ell-r)^2}{2s} }
\end{align*}
where we used the fact that $\ex{Y'_s} = \ell + 2\delta/3 s$ and a Chernoff-type bound, see for example \cite[Lemma A.4]{HL09}.
\end{proof}

We now move to the second step of the proof where we analyze the waiting times $W_1 + \dots + W_S$. Recall this is the total waiting time before the node $r = (3/4-\delta)n$ is reached.
\begin{lemma}
\label{lem:bound-waiting-time}
We have
\[
\pr{S \leq s, W_1 + \dots + W_S > c n \log^2 n} \leq \frac{1}{(3(1-8\eta))^{\ell}\binom{n}{\ell}} \cdot\frac{1}{\poly(n)}
\]
\end{lemma}
\begin{proof}
The techniques we use are similar to the techniques in \cite{HL09}, but we need to improve the analysis in several places. We try to use notation of \cite{HL09} as much as possible.

As in the proof of \cite[Lemma A.11]{HL09}, we start by defining the good event
\[
\eventfont{H} = \bigcap_{x=1}^n \event{\sum_{k=1}^S \1(Y_k \leq x) \leq \gamma x/\mu},
\]
where $\mu = 2c'\delta$.\footnote{We use this notation to apply \cite[Lemma A.5]{HL09} later.  $\mu$ corresponds to the probability of going forward minus the probability of going backward for a simplified walk that moves forward at most as fast as $Y_k$. In our case, we have $\mu > 2\delta/3$ because we stop after reaching state $r=(3/4-\delta)n$, and the probability of moving forward at $r$ is at least $1/2+\delta/3$.} The parameter $\gamma$ is going to be chosen later. This event is saying that states with small labels are not visited too many times. Later in the proof, we show that the $\pr{\eventfont{H}^c}$ is small.

Define the random variable $M = \min_{1 \leq i \leq S} Y_i$. We have
\begin{align}
\pr{W_1 + \dots + W_S > t, S \leq s, \eventfont{H}} &= \sum_{m=1}^{\ell} \pr{M=m, S \leq s, W_1 + \dots + W_S > t, \eventfont{H}} \notag \\
				&= \sum_{m=1}^{\ell} \pr{M=m} \pr{S \leq s, W_1 + \dots + W_S > t, \eventfont{H} | M=m} \notag \\
				&\leq \sum_{m=1}^{\ell} \pr{M \leq m}\max_{ \{y_i\} \text{ satisfying } M=m \text{ and } \eventfont{H} \text{ and } S \leq s} \pr{ W(y_1) + \dots + W(y_s) \geq t }, \label{eq:decomp-min}
\end{align}
where the maximum is taken over all sequences $y_1, \dots, y_s$ of possible walks and $W(y)$ is the waiting time at state $y$.

We bound $\pr{M \leq m}$ using Lemma \ref{lem:hitting-prob}. Recall that the random walk we are considering has transition probabilities that depend on the state we are in. More precisely, the probabilities of going from state $\ell$ to state $\ell+1$ is a decreasing function of $\ell$ for $\ell \leq r$. This makes it difficult to obtain a useful bound on $\pr{M \leq m}$ and so we are going to consider simplified walks for which $\pr{M \leq m}$ can only be greater. Note that at state $r$, the probability of moving to $r+1$ is $p_+(r) \geq 1/2+\delta/3$ (see \eqref{eq:r-drift}).

Define $q \eqdef \ceil{\frac{\log(n/\eta)}{\log(1+\delta)}}$. We handle the cases $\ell < q/\eta + 1$ and $\ell \geq q/\eta + 1$ separately. We start with $\ell \geq q/\eta + 1$.

We consider the following chain: the probabilities of moving forward between $\ell+1$ and $\ell+q$ are all set to $p_{+}(\ell+q)$, the value of this probability at state $\ell + q$. Moreover, for all states larger than $\ell+q$, we assign an equal probability of moving forward and backward. This defines a new walk to which we can apply Lemma \ref{lem:hitting-prob}. Assume for now that $\ell + q < r$. Using the same notation as in Lemma \ref{lem:hitting-prob}, we write $\alpha_{-} = \frac{p_{-}}{1-p_{-}} = 3 \cdot \frac{n-\ell}{\ell - 1}$, and $\alpha_q = \alpha_+(\ell+q) = \frac{p_{+}(\ell+q)}{1-p_{+}(\ell+q)}$, we obtain
\begin{align*}
\pr{M \leq \ell-1} &\leq \frac{1}{1+\alpha_{-} \frac{\alpha_q^q}{1+ \alpha_q^q + \dots + \alpha_q + 1 + \dots +1}} \\
&= \frac{1}{1+\alpha_- \frac{\alpha_q^q}{ \frac{\alpha_q^{q+1} - 1}{\alpha_q - 1} + (r-\ell-q-1)}}.
\end{align*}
We focus on the term involving $\alpha_q$:
\begin{align*}
\frac{\alpha_q^q}{ \frac{\alpha_q^{q+1} - 1}{\alpha_q - 1} + (r-\ell-q-1)} &= \frac{\alpha_q^q (\alpha_q - 1)}{ \alpha_q^{q+1} - 1 + (\alpha_q-1)(r-\ell-q-1)} \\
&\geq \frac{\alpha_q - 1}{\alpha_q} \cdot \frac{1}{1 + (\alpha_q -1) \frac{r - \ell - q - 1}{\alpha_q^{q+1}}}.
\end{align*}
We know that $\alpha_q \geq \alpha_+(r) \geq 1+2\delta$ using property \eqref{eq:r-drift} and as a result the previous expression is lower bounded by $(1-1/\alpha_q) (1-\eta)$. Continuing, we get
\begin{align*}
\pr{M \leq \ell -1} &\leq \frac{1}{1+ (1-\eta) \cdot \alpha_- \cdot  (1 - \frac{1}{\alpha_q})} \\
&= \frac{1}{1+ (1-\eta) \alpha_- - (1-\eta) \frac{\alpha_-}{\alpha_q}}.
\end{align*}
We now bound the quotient $\alpha_-/\alpha_q$.
\begin{align*}
\frac{\alpha_-}{\alpha_q} &= \frac{n-\ell}{\ell - 1} \frac{\ell+q-1}{n - (\ell+q)} \\
&= \left(1 + \frac{q}{\ell - 1}\right) \left(1 + \frac{q}{n-\ell-q}\right).
\end{align*}
We have $\frac{q}{n-\ell-q} \leq \frac{q}{n/4 - q} \leq \eta$ for large enough $n$. Moreover, by the assumption that $\ell \geq q/\eta + 1$, we have
\begin{align*}
\pr{M \leq \ell -1} &\leq \frac{1}{1 + (1-\eta) \alpha_- - (1-\eta) (1-\eta)^2} \\
		&\leq \frac{1}{(1- 8 \eta) \alpha_-}.
\end{align*}
This means that provided $q/\eta + 1 \leq \ell < r - q$, we have
\begin{align*}
\pr{M \leq \ell -1} &\leq \frac{1}{(1- 8 \eta)} \cdot \frac{1}{3} \frac{\ell - 1}{n - \ell}.
\end{align*}
Observe that if we have $\ell + q \geq r$, then we simply replace in the previous calculation $\ell+q$ with $r$ and the previous bound still holds in this case. To obtain a bound on $\pr{M \leq m}$ for $m < \ell - 1$, note that reaching $\ell-2$ before $r$ means reaching $\ell-1$ before $r$ starting at $\ell$ and reaching $\ell-2$ before $r$ starting at $\ell-1$, and these parts of the walk are independent. As a result, by induction, we have for $m \geq q/\eta + 1$,
\begin{align}
\pr{M \leq m} &\leq \frac{1}{(1-8\eta)^{\ell-m} 3^{\ell - m}} \cdot \frac{(\ell-1) (\ell - 2) \cdots m}{(n-\ell) (n-\ell+1) \cdots (n-m-1)} \notag \\
&\leq  \frac{1}{((1-8\eta)3)^{\ell}} \cdot \frac{\ell !}{n (n-1) \cdots (n-\ell+1)} \cdot \frac{3^m}{\ell (n-\ell)} \cdot \frac{n (n-1) \cdot (n-m)}{(m-1)!} \notag \\
&\leq \frac{1}{\left(3(1-8\eta)\right)^{\ell} \binom{n}{\ell}} \cdot (3 n)^m. \label{eq:bound-prob-m-2}
\end{align}
Note that whenever $m \leq q/\eta+1$, we can use the bound
\[
\pr{M \leq m} \leq \pr{M \leq q/\eta+1} \leq  \frac{1}{\left(3(1-8\eta)\right)^{\ell} \binom{n}{\ell}} \cdot (3 n)^{q/\eta+1}.
\]


We now look at the term $\max_{ \{y_i\} \text{ satisfying } M=m \text{ and } \eventfont{H} \text{ and } S \leq s} \pr{ W(y_1) + \dots + W(y_s) \geq t }$. As argued in the proof of \cite[Lemma A.11]{HL09}, the maximum is achieved when we make the walk visit as many times as possible the states with smaller labels. This means state $m$ is visited $\gamma m/\mu$ times, and all $i > m$ are visited $\gamma/\mu$ times. To avoid making the notation heavy, we assume that $\gamma/\mu$ is an integer. So we can write
\[
W(y_1) + \dots + W(y_s) \leq \sum_{i=1}^{\gamma m/\mu} G_{m,i} + \sum_{i=1}^{\gamma/\mu} \sum_{k=m+1}^{r} G_{k, i},
\]
where $G_{k,i}$ has a geometric distribution with parameter $2k/5n$ and the random variables $\{G_{k,i}\}$ are independent.
We are going to give upper tail bounds on the right hand side by computing the moment generating function. For any $\lambda \geq 0$, we have, using the moment generating function of a geometric distribution and the independence of the random variables:
\begin{align*}
\ex{\exp{\lambda\left(\sum_{i=1}^{\gamma m/\mu} G_{m, i} + \sum_{i=1}^{\gamma/\mu} \sum_{k=m+1}^{r} G_{k, i}\right)}} &= \left( \frac{ 2m/5n }{ e^{-\lambda} - 1 + 2m/5n} \right)^{\gamma m / 2} \prod^{r}_{k=m+1} \left( \frac{2k/5n}{e^{-\lambda} - 1 +2k/5n} \right)^{\gamma/\mu}.
\end{align*}
Now take $\lambda$ so that $e^{\lambda} = \frac{1}{1-m/(5n)}$. This leads to
\begin{align*}
\ex{\exp{\lambda\left(\sum_{i=1}^{\gamma m/\mu} G_{m, i} + \sum_{i=1}^{\gamma/\mu} \sum_{k=m+1}^{r} G_{k, i}\right)}}
&= \left(\frac{2m}{2m - m}\right)^{\gamma m/\mu} \cdot \prod^{r}_{k=m+1} \left( \frac{2k}{2k - m} \right)^{\gamma/\mu} \\
&\leq 2^{\gamma m/\mu} \left(\prod_{k=m+1}^{r} e^{\frac{m/2}{k-m/2}} \right)^{\gamma/\mu} \\
&\leq 2^{\gamma m/\mu} \left(e^{m/2 \cdot \ln n} \right)^{\gamma/\mu}.
\end{align*}
As a result, using Markov's inequality, we obtain
\begin{align*}
\pr{\sum_{i=1}^{\gamma m/\mu} G_{m,i} + \sum_{i=1}^{\gamma/\mu} \sum_{k=m+1}^{r} G_{k, i} > t}
&= \pr{\exp{\lambda\left( \sum_{i=1}^{\gamma m/\mu} G_{m,i} + \sum_{i=1}^{\gamma/\mu} \sum_{k=m+1}^{r} G_{k, i} \right) } > e^{\lambda t}}  \\
&\leq 2^{\gamma m/\mu} e^{\gamma m/(2\mu) \cdot \ln n} \cdot (1-m/(5n))^{t} \\
&\leq 2^{\gamma m/\mu} e^{\gamma m/(2\mu) \cdot \ln n} \cdot e^{-tm/(5n)}.
\end{align*}
Getting back to equation \eqref{eq:decomp-min}, we have using \eqref{eq:bound-prob-m-2}:
\begin{align*}
\pr{W_1 + \dots W_S > t, S \leq s, \eventfont{H}}
&\leq \sum_{m=1}^{\ell} \pr{M \leq m}  2^{\gamma m/\mu} e^{\gamma m/(2\mu) \cdot \ln n} \cdot e^{-tm/(5n)}  \\
&\leq \frac{1}{(1-8\eta)^\ell 3^{\ell} \binom{n}{\ell}} \cdot (3n)^{q/\eta+1} \sum_{m=1}^{q/\eta}   2^{\gamma m/\mu} e^{\gamma m/(2\mu) \cdot \ln n} \cdot e^{-tm/(5n)}  \\
& + \frac{1}{(1-8\eta)^\ell 3^{\ell} \binom{n}{\ell}} \cdot \sum_{m=q/\eta + 1}^{\ell}  (3n)^m  2^{\gamma m/\mu} e^{\gamma m/(2\mu) \cdot \ln n} \cdot e^{-tm/(5n)}  \\
&\leq \frac{1}{(1-8\eta)^\ell 3^{\ell} \binom{n}{\ell}} \cdot (1+(3n)^{q/\eta}) \cdot \sum_{m=1}^{\ell} \left( 3n 2^{\gamma/\mu} e^{\gamma/(2\mu) \cdot \ln n} \cdot e^{-t/(5n)} \right)^m
\end{align*}
Recall that $q = O(\log n)$ and thus if $t > c n \log^2 n$ with sufficiently large $c$, this probability is bounded by $O\left(\frac{1}{((1-8\eta)3)^{\ell}\binom{n}{\ell}} \cdot \frac{1}{\poly(n)}\right)$.

\comment{Actually, we could even do $2^{c \log^2 n}$ instead of $\poly(n)$.}

It now remains to bound $\pr{\eventfont{H}^c, S \leq s}$. Fix $x \in \{1, \dots, n\}$, we have
\begin{align*}
\pr{\sum_{k=1}^S \1(Y_k \leq x) > \gamma x/\mu, S \leq s}
&\leq \sum_{j=1}^s \pr{Y_j = x, \event{\forall i < j, Y_i > x}, j < S, \sum_{k=1}^S \1(Y_k \leq x) > \gamma x/\mu  } \\
&\leq \sum_{j=1}^s \pr{Y_j = x, \event{\forall i < j, Y_i > x}, j < S, \sum_{k=j+1}^{S_j} \1(Y_k \leq x) \geq \gamma x/\mu  } \\
&\leq \sum_{j=1}^s \pr{M \leq x} \cdot \pr{\sum_{k=j+1}^{S_j} \1(Y_k \leq x) \geq \gamma x/\mu  |Y_j = x, j < S},
\end{align*}
where we defined $S_j = \min\{ s \geq j + 1 : Y_s \geq r\}$.
To obtain the last inequality, we simply used the fact that $\event{Y_j = x, j < S} \subseteq \event{M \leq x}$. Moreover, $\event{j < S}$ can be determined by looking at $Y_1, \dots, Y_j$ and thus conditioned on $\event{Y_j=x}$, $Y_k$ for $k \geq j+1$ and also $S_j$ are independent of $\event{j < S}$. This means that we can drop $\event{j < S}$ from the conditioning.

To bound $\pr{M \leq x}$, we use \eqref{eq:bound-prob-m-2}. We can also bound $Y_k$ by a simpler random walk $Y'_k$ that moves forward with probability $1/2+\delta/3$, as we did in the proof of Lemma \ref{lem:bound-prob-s}. Thus, we obtain
\begin{align*}
\pr{\sum_{k=1}^S \1(Y_k \leq x) > \gamma x/\mu, S \leq s}
&\leq \frac{1}{((1-8\eta)3)^{\ell} \binom{n}{\ell}} \left((3n)^x + (3n)^{q/\eta} \right) \cdot s \cdot \pr{\sum_{k=1}^{\infty} \1(Y'_k \leq x) \geq \gamma x/\mu | Y'_0 = 0} \\
&\leq \frac{1}{((1-8\eta)3)^{\ell} \binom{n}{\ell}} \left((3n)^x + (3n)^{q/\eta} \right) \cdot s \cdot 2 \exp{-\frac{\mu (\gamma - 2) x}{2}},
\end{align*}
where we used \cite[Lemma A.5]{HL09}.
As a result, by a union bound,
\begin{align*}
\pr{\eventfont{H}^c, S \leq s} &\leq \frac{1}{((1-8\eta)3)^{\ell} \binom{n}{\ell}} \cdot 2s \cdot\left( \sum_{x=1}^n \exp{x\left(\log(3n) - \frac{\mu (\gamma - 2)}{2}\right)} + 3n^{q/\eta} \sum_{x=1}^n \exp{ - \frac{\mu (\gamma - 2)x}{2}}\right) \\
&\leq \frac{1}{((1-8\eta)3)^{\ell} \binom{n}{\ell}} \cdot \frac{1}{\poly(n)},
\end{align*}
where to get the last inequality, we choose $\gamma = c' \log n$ for large enough $c'$ and use the fact that  $s$ will be chosen linear in $n$.
Continuing, we reach
\begin{align*}
\pr{W_1 + \dots W_S > t, S \leq s}
&\leq \pr{W_1 + \dots W_S > t, S \leq s, \eventfont{H}} + \pr{\eventfont{H}^c, S \leq s} \\
&\leq \frac{1}{((1-8\eta)3)^{\ell} \binom{n}{\ell}} \frac{1}{\poly(n)}.
\end{align*}

We proved the desired bound when $\ell \geq q/\eta + 1$. It remains to deal with the case $\ell < q/\eta + 1$. We need to bound $\pr{M \leq \ell - 1}$ in a different way. For this we simply consider a walk that is even simpler than the one considered to obtain the bound in \eqref{eq:bound-prob-m-2}: let the probabilities of moving forward for all states above $\ell$ be equal to $p_+(r)$ which we know is at least $1/2+\delta/3$. Applying Lemma \ref{lem:hitting-prob}, we obtain
\[
\pr{M \leq \ell - 1} \leq \frac{1}{3 \delta} \cdot \frac{\ell - 1}{n - \ell},
\]
and then using the same argument as before
\[
\pr{M \leq m} \leq \frac{1}{(3\delta)^{\ell} \binom{n}{\ell}} \cdot (3\delta n)^m.
\]
Then we apply the exact same argument to obtain a bound
\begin{align*}
\pr{W_1 + \dots W_S > t, S \leq s} &\leq \frac{1}{(3\delta)^{\ell} \binom{n}{\ell}} \frac{1}{\poly(n)}.
\end{align*}
Now recall that $\ell < q/\eta + 1 = O(\log n)$ and thus for large enough $c$, we can make the term $1/\poly(n)$ be small enough to obtain the desired bound.
\end{proof}
To complete the proof of Lemma \ref{lem:time-reach-middle}, we just plug the bounds obtained from Lemma \ref{lem:bound-prob-s} with $s > 12n/\delta$ and from Lemma \ref{lem:bound-waiting-time} into equation \eqref{eq:t-s}.
\end{proof}

\begin{lemma}
\label{lem:time-reach-middle-+}
Let $\delta \in (0,1/16)$ and $\eta \in (0,1)$ be constants and $r_{+}$ satisfying condition \eqref{eq:r-drift}. Then for a large enough constant $c$ (depending on $\delta$ and $\eta$) and large enough $n$, we have  for any $\ell \geq r_{+}$
\[
\pr{T_{r_+}(\ell) > c n \log^2 n} \leq 2^{-2n}
\]
\end{lemma}
\begin{proof}
The proof is analogous to Lemma \ref{lem:time-reach-middle}, except that it is much easier to bound the waiting time. In fact, when $x > r_{+}$ we have $P(x,x) \leq 4/5$. This means that the waiting times $W_1, \dots, W_{S}$ can be assumed to have a geometric distribution with parameter $4/5$ and then proving a version of Lemma \ref{lem:bound-waiting-time} becomes a simple application of a Chernoff-type bound, and in fact one can obtain a better bound that is independent of $\ell$.
\end{proof}

\begin{proof}[of Theorem \ref{thm:mc-Qconvergence}]
Theorem \ref{thm:mc-Pconvergence} tells us that for $|\mu| = \ell$ and all $k \in \{1, \dots, n\}$,
\begin{equation}
\label{eq:Q-sum-k}
\sum_{|\nu| = k} Q^t(\mu,\nu) \leq 4^{\delta n} \frac{\binom{n}{k} 3^k }{4^n-1} + \frac{1}{(3-\eta)^{\ell} \binom{n}{\ell}} \frac{1}{\poly(n)}.
\end{equation}
Recall that there are exactly $\binom{n}{k}3^k$ distinct strings $\nu \in \{0,1,2,3\}^n$ such that $|\nu| = k$. We want to show that all these strings $\nu$ have basically the same value of $Q^t(\mu, \nu)$. 
For this, we view the chain $Q$ as a mixture of a part $\tilde{R}$ that can only mix the sites of the string without increasing its weight and a part $\tilde{Q}$ that can change the weight of the string. We then use invariance properties of these chains with respect to permuting the qubits and relabeling of nonzero elements $\{1,2,3\}$ to get the desired conclusion.

More precisely, Let $Z_t(\mu) \in \{0,1,2,3\}^n$ denote the state of the chain defined by $Q$ at step $t$ when started in the state $\mu$. From inequality \eqref{eq:Q-sum-k}, we can find an event $\eventfont{E}_P$ (in the notation of the proof of  of Theorem \ref{thm:mc-Pconvergence}, $\eventfont{E}_P = \event{T_{r_-} < t}$, see equation \eqref{eq:event-ep}) such that 
\[
\pr{\eventfont{E}^c_P} \leq \frac{1}{(3-\eta)^{\ell} \binom{n}{\ell} \poly(n)} \qquad \text{ and }
\qquad \pr{|Z_t(\mu)| = k, \eventfont{E}_P} \leq \frac{4^{\delta n} 3^k \binom{n}{k} }{4^n - 1},
\]
where $\eventfont{E}^c$ denotes the complement of the event $\eventfont{E}$.
This gives a natural candidate for the desired $p_{\delta}$, namely $p_{\delta}(\nu) = \pr{Z_t(\mu) = \nu, \eventfont{E}_P}$. The distance condition on $p_{\delta}$ is clearly satisfied: 
\[
\sum_{\nu \in \{0,1,2,3\}^n - \{0\}} \pr{Z_t(\mu) = \nu} - \pr{Z_t(\mu) = \nu, \eventfont{E}_P} = \pr{\eventfont{E}_P^c} \leq  \frac{1}{(3-\eta)^{\ell} \binom{n}{\ell} \poly(n)}.
\]
The objective of the remainder of the proof is to show that we have $p_{\delta}(\nu) \leq \frac{4^{2\delta n}}{4^n - 1}$.

Define $Q_{ij}$ to be the transition matrix of the Markov chain conditioned on having the gate act on qubits $i,j$. More precisely, $Q_{ij}(\mu,\nu) = 1$ if $\mu_i = \mu_j = 0$ and $\nu_i = \nu_j = 0$ and $Q_{ij}(\mu,\nu) = 1/15$ if $\mu_i \mu_j \neq 00$ and $\nu_i \nu_j \neq 00$, and all other entries of $Q_{ij}$ are zero. Thus, $Q(\mu, \nu) = \frac{1}{n(n-1)} \sum_{i \neq j} Q_{ij}(\mu, \nu)$. 
We now define
\begin{align*}
R_{ij}(\mu, \nu) = 
\left\{ \begin{array}{lc}
1 & \text{if } | \mu_i \mu_j |  = | \nu_i \nu_j | = 0 \\
1/3 & \text{if } | \mu_i \mu_j | = | \nu_i \nu_j | =  1, \mu_i = \nu_i = 0 \\
1/3 & \text{if } | \mu_i \mu_j | = | \nu_i \nu_j | =  1, \mu_j = \nu_j = 0 \\
1/9 & \text{if } | \mu_i \mu_j | = | \nu_i \nu_j | =  2.
\end{array} \right.
\end{align*}
and $\tilde{R}_{ij} = \frac{1}{2} R_{ij} + \frac{1}{2} \Pi_{ij} R_{ij}$ where $\Pi_{ij}$ simply swaps the coefficients at position $i$ and $j$. Also define
\begin{align*}
\tilde{Q}_{ij}(\mu, \nu) = 
\left\{ \begin{array}{lc}
1 & \text{if } | \mu_i \mu_j |  = | \nu_i \nu_j | = 0 \\
1/9 & \text{if } | \mu_i \mu_j | = 1 \text{ and } | \nu_i \nu_j | =  2 \\
2/3 \cdot 1/6 & \text{if } | \mu_i \mu_j | = 2 \text{ and } | \nu_i \nu_j | =  1 \\
1/3 \cdot 1/9 & \text{if } | \mu_i \mu_j | = 2 \text{ and } | \nu_i \nu_j | =  2.
\end{array} \right.
\end{align*}
It is simple to see that $Q_{ij} = \frac{2}{5} \tilde{R}_{ij} + \frac{3}{5} \tilde{Q}_{ij}$. We can then define $\tilde{R} = \frac{1}{n(n-1)} \sum_{i \neq j} \tilde{R}_{ij}$ and $\tilde{Q} = \frac{1}{n(n-1)} \sum_{i \neq j} \tilde{Q}_{ij}$ so that
\[
Q = \frac{2}{5} \tilde{R} + \frac{3}{5} \tilde{Q}.
\] 
Note that $\tilde{R}$ does not change the weight of any strings, but only performs swaps and locally randomizes $1$, $2$ and $3$.  
An important observation that will allow us to study $\tilde{R}$ and $\tilde{Q}$ independently is that $\tilde{R}\tilde{Q} = \tilde{Q}\tilde{R}$. In order to see this, observe first that $\tilde{R}_{ij} \tilde{Q}_{ij} = \tilde{Q}_{ij} = \tilde{Q}_{ij} \tilde{R}_{ij}$.  Also $\tilde{R}_{ij}$ and $\tilde{Q}_{i'j'}$ clearly commute if $\{i,j\} \cap \{i',j'\} = \emptyset$. Now for $j \neq j'$, we have $R_{ij} \tilde{Q}_{ij'} = \tilde{Q}_{ij'} R_{ij}$. However, $\Pi_{ij} R_{ij}$ does not commute with $\tilde{Q}_{ij'}$. But we can still write $\Pi_{ij} R_{ij} \tilde{Q}_{ij'} = R_{ij} \Pi_{ij} \tilde{Q}_{ij'} = R_{ij} \tilde{Q}_{jj'} \Pi_{ij} = \tilde{Q}_{jj'} \Pi_{ij} R_{ij}$. As a result,
\begin{align*}
\tilde{R} \tilde{Q} &= \frac{1}{n^2 (n-1)^2} \sum_{i \neq j, i' \neq j'}  \tilde{R}_{ij} \tilde{Q}_{i'j'} \\
&= \frac{1}{n^2 (n-1)^2} \left( \sum_{i \neq j, i' \neq j', |\{i,j\} \cap \{i',j'\}| \in \{0,2\}}  \tilde{Q}_{i'j'} \tilde{R}_{ij}  + 4 \sum_{i \neq j, i' \neq j', j \neq j'} \frac{1}{2} R_{ij} \tilde{Q}_{ij'} + \frac{1}{2} \Pi_{ij} R_{ij} \tilde{Q}_{ij'} \right) \\
&= \frac{1}{n^2 (n-1)^2} \left( \sum_{i \neq j, i' \neq j', |\{i,j\} \cap \{i',j'\}| \in \{0,2\}}  \tilde{Q}_{ij} \tilde{R}_{i'j'}  + 4 \sum_{i \neq j, i \neq j', j \neq j'} \frac{1}{2} \tilde{Q}_{ij'} R_{ij}  + \frac{1}{2} \tilde{Q}_{jj'} \Pi_{ij} R_{ij} \right) \\
&= \tilde{Q} \tilde{R}.
\end{align*}
The factor $4$ in the second line is to take into account the four possibilities $i=i'$, $i = j'$, $j = i'$ and $j = j'$.
As a result, for any $t \geq 1$, we can write $Q^t$ as
\begin{align}
\label{eq:decomposition-Q}
Q^t = \sum_{t_1+t_2=t} \left(\frac{3}{5}\right)^{t_1} \left(\frac{2}{5}\right)^{t_2}  \binom{t}{t_1} \tilde{R}^{t_2}\tilde{Q}^{t_1}.
\end{align}
Using equation \eqref{eq:decomposition-Q}, we see that $Z_t(\mu)$ can be generated as follows. Choose $T_1$ according to a binomial distribution with parameters $t$ and $3/5$ and run the chain $\tilde{Q}$ on $\mu$ for $T_1$ steps. Let $Z^{w}(\mu) \in \{0,1,2,3\}^n$ denote the state obtained at this time.  
Then, in the second phase, run the chain $\tilde{R}$ for $t-T_1$ steps obtaining the state $Z_t(\mu)$. Note that we have $|Z^{w}(\mu)| = |Z_{t}(\mu)|$.

We start with the case $k \leq \delta_0 n$ for some $\delta_0$ to be chosen later. Using \eqref{eq:B2-1} and \eqref{eq:binent-sqrt}, we have $\binom{n}{k} \leq 2^{n \binent(k/n)} \leq 2^{2\sqrt{\delta_0} n}$ and thus
\begin{align*}
\pr{Z_t(\mu) = \nu, \eventfont{E}_P} \leq \pr{|Z_t(\mu)| = |\nu|, \eventfont{E}_P} &\leq \frac{4^{\delta n} 3^k \binom{n}{k}}{4^n-1} \\
&\leq   \frac{4^{\delta n} 3^{\delta_0 n} 2^{2 \sqrt{\delta_0} n}}{4^n-1}.
\end{align*}
By choosing $\delta_0$ appropriately small, we obtain the desired result.

Now we assume that $\delta_0 < k < (1-\delta_0) n$. We deal with the case $k \geq (1-\delta_0) n$ at the end of the proof.
Note first that we have
\begin{align*}
\sum_{|\nu| = k} \pr{Z_t(\mu) = \nu, \eventfont{E}_P} \leq \frac{4^{\delta n} 3^k \binom{n}{k}}{4^n-1}.
\end{align*}

Our objective is to show that this total probability is basically evenly spread among all the $\nu$'s of weight $k$. For this, we condition on the value of $Z^{w}(\mu)$.
\begin{align}
\pr{Z_t(\mu) = \nu, \eventfont{E}_P} = \sum_{|\nu^{w}| = k}  \pr{Z^{w}(\mu) = \nu^w, \eventfont{E}_{P}} \cdot \pr{Z_t(\mu) = \nu | Z^{w}(\mu) = \nu^{w},  \eventfont{E}_{P}}. \label{eq:condition-zw}
\end{align}
Note that the event $\eventfont{E}_P$ only depends on the set of weights visited by the chain. As a result, by the Markov property for the second phase, the random variable $Z_t(\mu)$ is independent of $\eventfont{E}_P$ conditioned on $Z^w(\mu)$. In other words, $\pr{Z_t(\mu) = \nu | Z^{w}(\mu) = \nu^{w},  \eventfont{E}_{P}} = \pr{Z_t(\mu) = \nu | Z^{w}(\mu) = \nu^{w}}$. In order to evaluate this term, we study Markov chain for the second phase which is governed by the matrix $\tilde{R}$.

More precisely, we study the evolution of the support of $Z_s(\mu)$ for $s \geq T_1$ relative to the support of $Z^w(\mu)$. Define $I_s = |\supp(Z_s(\mu)) \cap \supp(Z^w(\mu))|$ for $s \geq T_1$. Recall that we have $|Z_s(\mu)|=|Z^{w}(\mu)| = k$ and thus the expected size for $\supp(Z_s(\mu)) \cap \supp(Z^w(\mu))$ if $\supp(Z_s(\mu))$ were completely random is $k^2/n$. 

It is simple to compute the transition probabilities of the chain $\{I_s\}_s$:
\begin{align*}
\pr{I_{s+1} = I_s +1} &= \frac{(k-I_s)^2}{n (n-1)} \\
\pr{I_{s+1} = I_s - 1} &= \frac{I_s(n - 2k + I_s)}{n (n-1)} \\
\pr{I_{s+1} = I_s} &= 1 - \pr{I_{s+1} = I_s +1} - \pr{I_{s+1} = I_s -1}.
\end{align*}
We can verify (for example by writing detailed balance equations) that the stationary distribution for this chain is given by $p_I(k') = \frac{\binom{k}{k'} \binom{n-k}{k-k'}}{\binom{n}{k}}$ for $k' \in \{0, \dots, k\}$. This allows us to bound the probability of reaching the state $k'$ when starting in a state $r'$, as was done for the chain $\{X_t\}$ in \eqref{eq:returnfn-from-middle}. This bound gets closer to the stationary probability $p_I(k')$ as $r'$ gets closer to $\frac{k^2}{n}$. More precisely, if $I_s(r')$ denotes the size of the intersection of the supports at step $s$ given that the starting state has an intersection size of $r'$, we have
\begin{align*}
\pr{I_s(r') = k'} \leq \frac{\binom{n}{k}}{\binom{k}{r'} \binom{n-k}{k-r'}} \cdot \frac{\binom{k}{k'} \binom{n-k}{k-k'}}{\binom{n}{k}}.
\end{align*}
We introduce the ``good'' event that for some $s \in [T_1, t]$, the walk $I_s$ gets close to the state $k^2/n$: $\eventfont{E}_I = \event{\exists s \in [T_1, t] : \frac{k^2}{n} - \delta_2 n \leq I_s \leq \frac{k^2}{n} + \delta_2 n}$.
Note that if $|r' - k^2/n| \leq \delta_2 n$, then using \eqref{eq:B2-1} and 
\begin{align*}
\binom{k}{r'} \binom{n-k}{k-r'} &\geq \frac{1}{n^2} 2^{k \cdot \binent(\frac{r'}{k}) + (n-k) \cdot \binent(\frac{k-r'}{n-k})} \\
					&\geq \frac{1}{n^2} 2^{n \cdot \binent(\frac{k}{n}) - n \binent(\frac{\delta_1}{\delta_0})} \\
					&\geq \frac{2^{-n \binent(\frac{\delta_1}{\delta_0})}}{n^2} \binom{n}{k}.
\end{align*}
For the second line, we used inequality \eqref{eq:binent-delta} which implies that $\binent(\frac{r'}{k}) \geq \binent(\frac{k}{n}) - \binent(\frac{\delta_2 n}{k}) \geq \binent(\frac{k}{n}) - \binent(\frac{\delta_2}{\delta_0})$, and similarly $\binent(\frac{k-r'}{n-k}) \geq  \binent(\frac{k}{n}) - \binent(\frac{\delta_2 n}{n-k}) \geq \binent(\frac{k}{n}) - \binent(\frac{\delta_2}{\delta_0})$. This means that we have
\begin{equation}
\label{eq:prob-r'-k'}
\pr{I_s(r') = k'} \leq n^2 2^{n \binent(\frac{\delta_2}{\delta_0})} \cdot \frac{\binom{k}{k'} \binom{n-k}{k-k'}}{\binom{n}{k}}
\end{equation}
whenever $|r' - k^2/n| \leq \delta_2 n$. Getting back to equation \eqref{eq:condition-zw}, we can write
\begin{align*}
\pr{Z_t(\mu) = \nu | Z^{w}(\mu) = \nu^{w}} &= \pr{\eventfont{E}_I^c} + \pr{Z_t(\mu) = \nu, \eventfont{E}_I | Z^{w}(\mu) = \nu^{w}}.
\end{align*}
We start by bounding $\pr{\eventfont{E}_I^c}$. For this, observe that for $\pr{I_{s+1} = I_s + 1} - \pr{I_{s+1} = I_s - 1} = \frac{k^2 - I_s n}{n(n-1)}$. This means that if $I_s \geq \frac{k^2}{n} + \delta_2 n$, there is a $\delta_2$ negative drift, and similarly there is a constant positive drift if $I_s \leq \frac{k^2}{n} - \delta_2 n$. Using standard methods, one can conclude that $\pr{\eventfont{E}^c_I | t- T_1 \geq n \log n} \leq 2^{-10n}$. In addition for large enough $t$, $\pr{t - T_1 \geq n \log n} \geq 1 - 2^{-10n}$. \footnote{Note that having $T_1 \geq c' n$ for some large enough constant $c'$ depending on $\delta$ would be good enough; we choose $n \log n$ simply to avoid introducing additional constants.} Then one can directly conclude $\pr{\eventfont{E}^c_I} \leq \pr{t-T_1 < n \log n} +  \pr{\eventfont{E}^c_I | t- T_1 \geq n \log n} \leq 2 \cdot 2^{-10n}$. 

We can write
\begin{align}
\sum_{|\nu'| = k : |\supp(\nu') \cap \supp(\nu^{w})| = k'} \pr{Z_t(\mu) = \nu', \eventfont{E}_I | Z^{w}(\mu) = \nu^{w}} &\leq \max_{|r'-k^2/n| \leq \delta_2 n} \max_{s} \pr{I_s(r') = k'} \\
&\leq n^2 2^{n \binent(\frac{\delta_2}{\delta_0})} \cdot \frac{\binom{k}{k'} \binom{n-k}{k-k'}}{\binom{n}{k}}.
\label{eq:sum-k-k'}
\end{align}
Now it remains to say that many of the terms in this sum are actually the same. For this, we use invariance properties of $\tilde{R}$.

Under all permutations $\pi \in \mfS_n$ of $\{1, \dots, n\}$, and all functions $\gamma \in (\mfS_3)^{\times n}$ that permute the Pauli operators $\{1,2,3\}$ on each qubit, we have
\begin{equation}
\label{eq:invariance}
\tilde{R}((\pi \circ \gamma )(\mu), (\pi \circ \gamma )(\nu))=\tilde{R}(\mu,\nu).
\end{equation}
It follows that $\tilde{R}(\mu, (\pi_0 \circ \gamma_0)(\nu)) = \tilde{R}((\pi_0 \circ \gamma_0)(\mu), (\pi_0 \circ \gamma_0)(\nu)) = \tilde{R}(\mu, \nu)$ for any $\pi_0 \in \mfS_n$ and $\gamma_0 \in (\mfS_3)^{n}$ such that $\pi_0 \circ \gamma_0 (\mu) = \mu$, e.g., if $\pi_0$ and $\gamma_0$ act outside the support of $\mu$.

As a result, we have that $\pr{Z_t(\mu) = \nu | Z^{w}(\mu) = \nu^w} = \pr{Z_t(\mu) = \nu' | Z^{w}(\mu) = \nu^w}$ if $\nu'$ can be obtained from $\nu$ by a permutation and relabeling of the Pauli operators that act outside the support of $\nu^w$. If $|\supp(\nu) \cap \supp(\nu^w) | = k'$, then there are $3^{k-k'} \binom{n-k}{k-k'}$ distinct $\nu'$ that can be obtained in this way. 

Invariance of the transition probabilities under maps that act on the support of $\nu^w$ is slightly more complicated. For any permutation $\pi$ of the support of $\nu^w$, and any relabeling $\gamma_{\pi}$ that satisfies $\gamma_{\pi}(\nu) = \pi^{-1}(\nu)$, $\pi \circ \gamma_{\pi}$ keeps $\nu^w$ unchanged. Note that for any $\pi$ there is at least one such $\gamma_{\pi}$. This means that also $\nu' = \pi \circ \gamma_{\pi}(\nu)$ obtained in this way satisfy $\pr{Z_t(\mu) = \nu'| Z^{w}(\mu) = \nu^w} = \pr{Z_t(\mu) = \nu | Z^{w}(\mu) = \nu^w}$. By combining with invariants outside the support of $\nu^w$, we obtain a total of $3^{k-k'} \binom{n-k}{k-k'} \cdot \binom{k}{k'}$ distinct $\nu'$ for which $\pr{Z_t(\mu) = \nu' | Z^{w}(\mu) = \nu^w} = \pr{Z_t(\mu) = \nu | Z^{w}(\mu) = \nu^w}$. 

The total number of $\nu'$ such that $|\nu'| = k$ and $|\supp(\nu') \cap \supp(\nu^w)| = k'$ is $3^{k} \binom{n-k}{k-k'} \cdot \binom{k}{k'}$, so our objective is to prove that there are roughly $3^{k'}$ additional relabelings that keep the transition probability invariant. In particular, we want to show that relabelings acting on the support of $\nu^w$ keep this probability unchanged. For this we argue as in Appendix \ref{sec:app-one-design}, that with high probability, most of the sites are acted upon at least once in the second phase. More precisely introduce the event $\eventfont{E}_A$ that between times $T_1$ and $t$, a $(1-\delta_1)$ fraction of the sites $\{1, \dots, n\}$ are acted upon in at least one step. First, let us see that this event happens with high probability. In fact, by applying a union bound on all the subsets of size $\delta_1 n$, we directly get that for sufficiently large $n$, $\pr{\eventfont{E}_A^c | t - T_1 \geq n \log n} \leq 2^{-10n}$ and thus $\pr{\eventfont{E}_A^c} \leq 2 \cdot 2^{-10 n}$.

As argued in Appendix \ref{sec:app-one-design}, we can condition on the set of all sites that are acted upon in some step between $T_1$ and $t$. Then any string that is obtained from $\nu$ by applying a relabeling $\gamma$ that acts on these sites has the same probability as $\nu$. If this set of sites has size at least $(1-\delta_1)n$, i.e., the event $\eventfont{E}_A$ holds, there are at least $k'-\delta_1 n$ such sites that are in $\supp(\nu) \cap \supp(\nu^w)$. This means that under the event $\eventfont{E}_A$, there are at least $3^{k'-\delta_1 n}$ strings $\nu'$ obtained from $\nu$ by applying a relabeling on some sites of $\supp(\nu) \cap \supp(\nu^w)$. As a result, using \eqref{eq:sum-k-k'},
\begin{align*}
&\pr{Z_t(\mu) = \nu | Z^{w}(\mu) = \nu^w} \\
&\leq \frac{1}{3^{k'-\delta_1 n} 3^{k-k'} \binom{n-k}{k-k'} \cdot \binom{k}{k'}} \cdot \sum_{|\nu'| = k : |\supp(\nu') \cap \supp(\nu^w)| = k'} \pr{Z_t(\mu) = \nu', \eventfont{E}_A | Z^{w}(\mu) = \nu^{w}} + \pr{\eventfont{E}_A^c} \\
&\leq n^2 2^{n \binent(\frac{\delta_2}{\delta_0})} 3^{\delta_1 n} \cdot \frac{1}{3^k \binom{n}{k}} + 2 \cdot 2^{-10n}.
\end{align*}
Going back to \eqref{eq:condition-zw}, we obtain
\begin{align*}
\pr{Z_t(\mu) = \nu, \eventfont{E}_P} &\leq \frac{2 n^2 2^{n \binent(\frac{\delta_2}{\delta_0})} 3^{\delta_1 n}}{3^k \binom{n}{k}} \sum_{|\nu^w| = k} \pr{Z^w = \nu^w, \eventfont{E}_P} \\
&= \frac{2 n^2 2^{n \binent(\frac{\delta_2}{\delta_0})} 3^{\delta_1 n}}{3^k \binom{n}{k}} \pr{|Z_t(\mu)| = k, \eventfont{E}_P} \\
&\leq \frac{ 2 n^2 2^{n \binent(\frac{\delta_2}{\delta_0})} 3^{\delta_1 n} \cdot 4^{\delta n}}{4^n - 1} \\
&\leq \frac{16^{\delta n}}{4^n-1},
\end{align*}
for large enough $n$ and where in the last step we choose $\delta_1 > 0$ and $\delta_2 > 0$ small enough constants.

Now it only remains to handle the case $k \geq (1-\delta_0) n$. In this case, the size of the intersection $k' = |\supp(\nu) \cap \supp(\nu^w)| \geq 2k - n$. We then observe that on the event $\eventfont{E}_{A}$, we can obtain at least $3^{k'-\delta_1 n}$ distinct $\nu'$ such that $\pr{Z_t(\mu) = \nu' | Z^w(\mu) = \nu^w} = \pr{Z_t(\mu) = \nu | Z^w(\mu) = \nu^w}$. As a result
\begin{align*}
\pr{Z_t(\mu) = \nu | Z^{w}(\mu) = \nu^w}
&\leq \frac{1}{3^{k'-\delta_1 n}} \sum_{|\nu'| = k : |\supp(\nu') \cap \supp(\nu^w)| = k'} \pr{Z_t(\mu) = \nu' | Z^{w}(\mu) = \nu^{w}} \\
&\leq \frac{1}{3^{2k - n - \delta_1 n}} \sum_{|\nu'| = k} \pr{Z_t(\mu) = \nu' | Z^{w}(\mu) = \nu^{w}} \\
&\leq \frac{1}{3^{k -\delta_0 n - \delta_1 n}} \frac{4^{\delta n} 3^k \binom{n}{k}}{4^n - 1} \\
&\leq \frac{ 3^{(\delta_1 + \delta_0)n} 2^{h(\delta_0) n} 4^{\delta n}}{4^n - 1}.
\end{align*}
For small enough $\delta_0$ and $\delta_1$, this leads to the desired result.
\end{proof}

\section{Conclusion}

We proved that decoupling can be achieved using a number of two-qubit gates that is almost linear in the system size. This implies that information processing tasks that can be achieved via decoupling can be implemented with a circuit of almost linear size and polylogarithmic depth.

Our result also show that  a class of random time dependent Hamiltonians self-thermalize at a speed that is close to the signaling bound.  It is an interesting question if a similar result applies to the decoupling time for broader classes of two-body Hamiltonians on the complete graph, and whether decoupling can occur at a time scale close to the signaling bound of  $O(n^{1/d}) $ for interactions on $d$-dimensional lattices.

As far as optimality is concerned, it would be interesting to improve the depth to $O(\log n)$. For that, one would probably need to study directly a parallel random circuit model, as the parallelization step involves an additional $O(\log n)$ factor.


\section*{Acknowledgements}

We would like to thank Fernando Brandao, Patrick Hayden, David Poulin, Renato Renner, Lidia del Rio, Marco Tomamichel and Stephanie Wehner for helpful discussions and Aram Harrow for his comments. The research of WB is supported by the Centre de Recherches Math\'ematiques at the University of Montreal, Mprime, and the Lockheed Martin Corporation. The research of OF is supported by the European Research Council grant No. 258932.

\appendix

\section{A generalisation of the gambler's ruin lemma}
Consider a random walk on a line indexed from $-1$ to $a$. At positions $i > 0$, the probability of moving forward is $p_{+}(i)$ (depending on $i$) and for points $i \leq 0$, the probability of moving forward is $p_{-}$. The following lemma gives a bound on the probability of hitting the node $-1$ before hitting $a$ when starting at position $0$. In our setting, we are interested in the case where $p_{-}$ and $p_{+}$ are (significantly) larger than $1/2$ so that the probability of hitting $-1$ before $a$ is small.

\begin{lemma}
\label{lem:hitting-prob}
Assume $p_+(i), p_{-} > 1/2$. Then the probability of hitting $-1$ before $a$ is exactly
\[
\frac{1}{1+ \alpha_{-} \cdot \frac{\prod_{j=1}^{a-1} \alpha_+(j)}{1 + \sum_{i=1}^{a-1} \prod_{j=i}^{a-1} \alpha_+(j)} } \ ,
\]
where $\alpha_+(i) = \frac{p_{+}(i)}{1-p_+(i)}$ and $\alpha_- = \frac{p_-}{1-p_-}$. In particular, if $\alpha_+(i) = \alpha_+$ for all $i$, this probability becomes
\[
\frac{1}{1+ \alpha_{-} \cdot \frac{\alpha_+^{a} - \alpha_+^{a-1}}{\alpha_+^{a} - 1}  } \leq \frac{1}{1+ \alpha_{-} \cdot (1-1/\alpha_+)  } \ .
\]

\end{lemma}
\begin{proof}
Let $P_i$ be the probability of first reaching $-1$ when starting at position $i$. We can write for any 
for $i \in [1, a-1]$, $P_{i} = p_{+}(i) P_{i+1} + (1-p_{+}(i)) P_{i-1}$, which can be re-written as
\[
\frac{p_{+}(i)}{1-p_{+}(i)} \left(P_{i} - P_{i+1} \right) =  \left(P_{i-1} - P_i \right).
\]
We now use the boundary condition at node $a$: $P_a = 0$. Thus, $ \left(P_{a-2} - P_{a-1} \right) =  \frac{p_{+}(a-1)}{1-p_{+}(a-1)} P_{a-1}$. Moreover, we see by induction that for any $i \geq 1$, $P_{i-1} - P_i = \left( \prod_{j=i}^{a-1}\frac{p_+(j)}{1-p_+(j)}\right) P_{a-1}$. We can now write a telescoping sum
\[
P_{0} - P_{a-1} = \sum_{i=1}^{a-1} P_{i-1} - P_{i} = \sum_{i=1}^{a-1} \prod_{j=i}^{a-1}\alpha_+(j) \cdot P_{a-1}.
\]
As a result, 
\[
P_{0} = P_{a-1} \left( 1 + \sum_{i=1}^{a-1} \prod_{j=i}^{a-1}\alpha_+(j). \right).
\]
We can then write $P_{-1} - P_{0} = \frac{p_-}{1-p_-} \left( P_{0} - P_{1} \right) = P_{a-1} \cdot \prod_{j=1}^{a-1}\alpha_+(j) \cdot \frac{p_-}{1-p_-}$.

Now, we use our second boundary condition $P_{-1} = 1$. We have
\begin{align*}
1 = P_{-1} &= P_0 +  P_{a-1} \cdot \alpha_{-} \prod_{j=1}^{a-1}\alpha_+(j) \\
		&= P_0 \left(1 + \alpha_{-} \frac{\prod_{j=1}^{a-1}\alpha_+(j)}{\sum_{i=1}^{a-1} \prod_{j=i}^{a-1}\alpha_+(j)} \right),		
\end{align*}
which leads to the desired result.
\end{proof}

\section{Sequential random quantum circuits are approximate $1$-designs}
\label{sec:app-one-design}
The objective of this section is to show that we have for $t > c n \log n$,
\begin{equation}
\label{eq:one-design}
\exc{U_t}{\tr[\tilde{\cT}(U_t \rho_{AE} U^{\dagger}_t) \tilde{\tau}_B \otimes \tilde{\rho}_{E}]} \geq \left(1-\frac{1}{\poly(n)}\right) \tr[\tilde{\tau}_B^2] \tr[\tilde{\rho}_E^2].
\end{equation}

Let us generate the circuit $U_t$ by first choosing the pair of qubits $S = \{(i_1, j_1), \dots, (i_t, j_t)\}$ on which each of the $t$ gates act and then choosing the two-qubit unitaries $V_1, \dots, V_t$ that are applied in each time step. We then write $U_t = V_t(i_t, j_t) \cdots V_1(i_1,j_1)$. Let $\eventfont{G}$ be the event that $\{i_1, j_1, i_2, j_2, \dots, i_t, j_t\} = [n]$. It then follows that if we fix such an $S$ and take the expectation over the choice of $V_1, \dots, V_t$, we have for any $S$ that satisfies $\eventfont{G}$,
\[
\exc{V_1, \dots, V_t}{U_t \sigma_{\mu} U_t^{\dagger}} = 0,
\]
for all $\mu \neq 0$. As a result we have
\[
\exc{V_1, \dots, V_t}{\tr[\tilde{\cT}(U_t \rho_{AE} U^{\dagger}_t) \tilde{\tau}_B \otimes \tilde{\rho}_{E}]} = \tr\left[\tilde{\cT}\left(\frac{\id}{2^n}\right) \otimes \tilde{\rho}_E \tilde{\tau}_B \otimes \tilde{\rho}_E \right] =  \tr[\tilde{\tau}_B^2] \tr[\tilde{\rho}_E^2],
\]
for any fixed $S$ that satisfies $\eventfont{G}$.
Now it only remains to bound the probability of the event $\eventfont{G}^c$, which is the complement of $\eventfont{G}$. The probability that qubit $1$ is not affected by any gate is $(1-2/n)^t$. Then, by a union bound, we have $\pr{\eventfont{G}^c} \leq n (1-2/n)^t \leq n e^{2t/n} \leq \frac{1}{\poly(n)}$.

\section{Bounding the total mass of coefficients at a certain weight}

\begin{lemma}
\label{lem:level}
Let $\rho_{AE}$ be such that $\entHtwo(A|E)_{\rho} \geq -(1-\e)n$ with $\e > 0$, i.e.,
\[\tr[\tilde{\rho}_{AE}^2] \leq 2^{(1-\e)n},
\]
where $\tilde{\rho}_{AE} = \rho_E^{-1/4} \rho_{AE} \rho_E^{-1/4}$. Then, there exists $\eta > 0$ (depending only on $\e$) such that for all $\ell$,
\[
\sum_{\nu : |\nu| = \ell} \tr\left[ \tr_A[ \sigma_{\nu} \tilde{\rho}_{AE}]^2 \right] \leq 12n^4 \cdot (3-\eta)^{\ell} \binom{n}{\ell}
\]
\end{lemma}
\begin{proof}
Fix $m = \ceil{4\ell/3}$ and apply Theorem \ref{thm:min-entropy-sampling}, we obtain
\[
\exc{|S|=m}{\tr[\tilde{\rho}_{A_SE}^2]} \leq (n^2+1) \cdot 2^{(1-\delta)m}.
\]
But we know that
\begin{align*}
\sum_{S : |S| = m} \tr[\tilde{\rho}_{A_SE}^2] &= \sum_{S : |S| = m} \frac{1}{2^{m}} \sum_{\nu \in \{0,1,2,3\}^S} \tr[ \tr_A[\sigma_{\nu}\tilde{\rho}_{AE}]^2] \\
&\geq \frac{1}{2^m} \sum_{\nu \in \{0,1,2,3\}^n : |\nu| = \ell} \binom{n-\ell}{m-\ell}\tr[ \tr_A[\sigma_{\nu}\tilde{\rho}_{AE}]^2],
\end{align*}
by simply forgetting the terms $\tr[ \tr_A[\sigma_{\nu}\tilde{\rho}_{AE}]^2]$ for which $|\nu| \neq \ell$. Note that $\binom{n-\ell}{m-\ell}$ is the number of sets $S$ of size $m$ in which the support of $\nu$ is included. As a result, we have
\begin{align*}
\sum_{\nu: |\nu| = \ell} \tr[ \tr_A[\sigma_{\nu} \tilde{\rho}_{AE}]^2 ] &\leq \frac{2^m}{\binom{n-\ell}{m-\ell}} \cdot \binom{n}{m} (n^2 + 1) 2^{(1-\delta)m}\\
				&= (n^2+1) \binom{n}{\ell}\frac{4^m}{\binom{m}{\ell}} 2^{-\delta m}.
\end{align*}
To conclude, we note that $3^{\ell} \binom{m}{\ell} \geq 3^{\ell} \frac{2^{m \binent(3/4)}}{m (m+1)} \geq \frac{3^{3/4m-1} 2^{m \binent(3/4)}}{n(n+1)}$,
where $\binent$ is the binary entropy function.
We conclude that
\begin{align*}
\sum_{\nu: |\nu| = \ell} \tr[ \tr_A[\sigma_{\nu} \tilde{\rho}_{AE}]^2 ]				
  &= 3 (n^2+1) 3 n (n+1) \binom{n}{\ell} 3^{\ell} 2^{-\delta m} \\
  &\leq 12 n^4 \binom{n}{\ell} (3-\eta)^{\ell}
  \end{align*}
  for an appropriate choice of constant $\eta > 0$.
\end{proof}

\begin{theorem}[Fully quantum entropy sampling \cite{DFW13}]
\label{thm:min-entropy-sampling}
Let $\rho_{AE}$ be such that $\entHtwo(A|E)_{\rho} \geq -(1-\e)n$ with $\e > 0$, i.e.,
\[\tr[\tilde{\rho}_{AE}^2] \leq 2^{(1-\e)n},
\]
where $\tilde{\rho}_{AE} = \rho_E^{-1/4} \rho_{AE} \rho_E^{-1/4}$. Then, there exists $\delta > 0$ (depending only on $\e$) such that for all $m$, when taking the average over all subsets $S$ of size $m$,
\[
\exc{|S|=m}{\tr[\tilde{\rho}_{A_SE}^2]} \leq (n^2+1) \cdot 2^{(1-\delta)m}.
\]
\end{theorem}

\section{Properties of binomials}

We use $\binent$ to denote the binary entropy function $\binent(\alpha) = -\alpha \log(\alpha) - (1-\alpha) \log(1-\alpha)$. We use the following simple estimates for binomial coefficients (see~\cite[Lemma 9.2]{MU05}).
Let $\alpha \in [0,1]$ such that $\alpha n$ is an integer. Then
\begin{equation}
\label{eq:B2-1}
\sum_{k=0}^{\alpha n} \binom{n}{k} \leq 2^{n \binent(\alpha)},
\end{equation}
and
\begin{equation}
\label{eq:B2-2}
\frac{2^{n \binent(\alpha)}}{n+1} \leq \binom{n}{\alpha n}.
\end{equation}
We also use
\begin{equation}
\label{eq:binent-delta}
|\binent(\alpha + \delta) - \binent(\alpha)| \leq h(\delta),
\end{equation}
for all $\alpha, \delta \geq 0$ with $\alpha+\delta \leq 1$. To prove this, we observe that $f: \alpha \mapsto \binent(\alpha+\delta) - \binent(\alpha)$ is a decreasing function of $\alpha \in [0, 1-\delta]$ and thus $|\binent(\alpha+\delta) - \binent(\alpha) | \leq \max(f(0), f(1-\delta)) = h(\delta)$.
Moreover,
\begin{equation}
\label{eq:binent-sqrt}
\binent(\alpha) \leq 2\sqrt{\alpha(1-\alpha)}.
\end{equation}

\bibliographystyle{alphaplus}
\bibliography{big}

\end{document}